\documentclass[11pt]{article}
\pdfoutput=1 
\usepackage[sc]{mathpazo}

\usepackage[margin=1in]{geometry}
\usepackage[english]{babel}
\usepackage[utf8x]{inputenc}
\usepackage[compact]{titlesec}
\usepackage{cmap}
\usepackage[T1]{fontenc}
\usepackage{bm}
\usepackage{enumerate}
\usepackage{booktabs}
\usepackage{mathtools}
\usepackage{amsfonts}
\usepackage{amsmath}
\usepackage{amssymb}
\usepackage{amsthm}
\usepackage{float}
\usepackage{graphics}
\usepackage{hyperref}
\usepackage[svgnames]{xcolor}
\hypersetup{colorlinks={true},urlcolor={blue},linkcolor={DarkBlue},citecolor=[named]{DarkGreen},linktoc=all}
\usepackage{natbib}
\usepackage{microtype}
\usepackage[capitalise,nameinlink,noabbrev]{cleveref}
\usepackage{doi}

\newcommand{\BibTeX}{\rm B\kern-.05em{\sc i\kern-.025em b}\kern-.08em\TeX}

\usepackage[labelfont={normalfont,bf},textfont=it]{caption}
\usepackage{subcaption}

\usepackage{nicefrac}
\usepackage{pifont}
\newcommand{\cmark}{\ding{51}}

\usepackage{apxproof}
\usepackage{array,multirow,graphicx,bigdelim}

\input{insbox}
\usepackage{tikz}
\usepackage{tikz-dimline}
\usepackage{tikz-cd}
\usetikzlibrary{fit,calc,math,shapes,shapes.multipart,decorations.text,arrows,decorations.markings,decorations.pathmorphing,shapes.geometric,positioning,decorations.pathreplacing}
\tikzset{snake it/.style={decorate, decoration=snake}}

\colorlet{mygray}{gray!40}

\makeatletter
\let\oldnl\nl
\newcommand{\nonl}{\renewcommand{\nl}{\let\nl\oldnl}}
\makeatother

\usepackage{thmtools,thm-restate}
\usepackage[capitalise,nameinlink,noabbrev]{cleveref}

\newtheorem{dfn}{Definition}
\newtheorem{definition}{Definition}
\newtheorem{example}{Example}
\newtheorem{lemma}{Lemma}

\newtheorem{claim}{Claim}
\Crefname{claim}{Claim}{Claims}

\Crefname{claim}{Claim}{Claims}
\Crefname{corollary}{Corollary}{Corollaries}
\Crefname{definition}{Definition}{Definitions}
\Crefname{example}{Example}{Examples}
\Crefname{lemma}{Lemma}{Lemmas}
\Crefname{property}{Property}{Properties}
\Crefname{proposition}{Proposition}{Propositions}
\Crefname{remark}{Remark}{Remarks}
\Crefname{theorem}{Theorem}{Theorems}

\allowdisplaybreaks

\def\mrk{\mathsf{Mark}}
\def\supp{\mathsf{Sup}}
\def\usupp{\mathsf{Un}\supp}

\newcommand{\A}{\mathcal{A}}
\newcommand{\C}{\mathcal{C}}
\newcommand{\CISP}{\textup{CISP}}
\newcommand{\EF}[1]{\ifstrempty{#1}{\textrm{\textup{EF}}}{\textrm{\textup{EF{$#1$}}}}}
\newcommand{\EFX}{\textrm{\textup{EFX}}}

\newcommand{\F}{\mathcal{F}}

\newcommand{\MMS}{\textrm{\textup{MMS}}}

\newcommand{\PO}{\textup{PO}}

\definecolor{DarkGreen}{rgb}{0.1,0.5,0.1}

\newcommand{\T}{\mathcal{T}}
\newcommand{\V}{\mathcal{V}}
\newcommand{\X}{\mathbf{X}}
\newcommand{\Y}{\mathbf{Y}}
\newcommand{\abs}[1]{\lvert #1 \rvert}

\usepackage[linesnumbered,lined,boxed,ruled,vlined]{algorithm2e}

\SetKwInOut{Parameters}{Parameters}
\SetKwComment{Comment}{$\triangleright$\ }{}
\SetAlFnt{\small}
\SetAlCapFnt{\small}
\SetAlCapNameFnt{\small}
\SetAlCapHSkip{0pt}
\IncMargin{-\parindent}

\theoremstyle{remark}

\title{Fair Interval Scheduling of Indivisible Chores}
\author{
	\begin{tabular}{m{0.12\textwidth}m{0.12\textwidth}m{0.12\textwidth}m{0.12\textwidth}m{0.12\textwidth}m{0.12\textwidth}}
		\multicolumn{2}{c}{\textbf{Sarfaraz Equbal}} & \multicolumn{2}{c}{\textbf{Rohit Gurjar}} & \multicolumn{2}{c}{\textbf{Yatharth Kumar}}\\
		\multicolumn{2}{c}{\small{IIT Bombay}} & \multicolumn{2}{c}{\small{IIT Bombay}} & \multicolumn{2}{c}{\small{IIT Delhi}}\\
		\multicolumn{2}{c}{\href{mailto:sequbal@cse.iitb.ac.in}{\small{\texttt{sequbal@cse.iitb.ac.in}}}} & \multicolumn{2}{c}{\href{mailto:rgurjar@cse.iitb.ac.in}{\small{\texttt{rgurjar@cse.iitb.ac.in}}}}
		& \multicolumn{2}{c}{\href{mailto:cs1200413@iitd.ac.in}{\small{\texttt{cs1200413@iitd.ac.in}}}}\\
		&&&&&\\
		\multicolumn{2}{c}{\textbf{Swaprava Nath}} & \multicolumn{2}{c}{\textbf{Raghuvansh R. Saxena}} & \multicolumn{2}{c}{\textbf{Rohit Vaish}}\\
		\multicolumn{2}{c}{\small{IIT Bombay}} & \multicolumn{2}{c}{\small{Tata Institute of Fundamental Research}} & \multicolumn{2}{c}{\small{IIT Delhi}}\\
		\multicolumn{2}{c}{\href{mailto:swaprava@cse.iitb.ac.in}{\small{\texttt{swaprava@cse.iitb.ac.in}}}} &\multicolumn{2}{c}{\href{mailto:raghuvansh.saxena@gmail.com}{\small{\texttt{raghuvansh.saxena@gmail.com}}}} & \multicolumn{2}{c}{\href{mailto:rvaish@iitd.ac.in}{\small{\texttt{rvaish@iitd.ac.in}}}}\\
	\end{tabular}
}

\date{}

\begin{document}

\maketitle

\begin{abstract}
We study the problem of fairly assigning a set of discrete tasks (or chores) among a set of agents with additive valuations. Each chore is associated with a start and finish time, and each agent can perform at most one chore at any given time. The goal is to find a fair and efficient schedule of the chores, where fairness pertains to satisfying \emph{envy-freeness up to one chore} (\EF{1}) and efficiency pertains to \emph{maximality} (i.e., no unallocated chore can be feasibly assigned to any agent). Our main result is a polynomial-time algorithm for computing an \EF{1} and maximal schedule for two agents under \emph{monotone} valuations when the conflict constraints constitute an \emph{arbitrary} interval graph. The algorithm uses a coloring technique in interval graphs that may be of independent interest. For an arbitrary number of agents with identical additive valuations, we show the existence of an EF1 and maximal schedule when the constraints constitute a path graph. This result uses a reduction to the ``cycle-plus-triangles'' theorem. Using different techniques, we provide an efficient algorithm for finding such a schedule when there are four or more agents and the valuations are further assumed to be dichotomous. We also show that stronger fairness and efficiency properties, including envy-freeness up to any chore (\EFX{}) along with maximality and \EF{1} along with Pareto optimality, cannot be achieved.
\end{abstract}

\newpage
\tableofcontents
\newpage

\section{Introduction}
\label{sec:Introduction}

Fair allocation of indivisible resources has become a significant area of study within economics, operations research, and computer science~\citep{BT96fair,BCE+16handbook,M19fair}. The main goal here is to distribute a set of discrete resources among agents with differing preferences in a way that provably ensures fairness and economic efficiency. This field has generated extensive theoretical and practical interest in recent years. On the theoretical side, various fairness notions have been developed and a wide range of algorithmic techniques have emerged~\citep{AAB+23fair}. In terms of practical applications, there are various areas of use, such as course allocation~\citep{BCK+17course}, public housing allocation~\citep{BCH+20price}, and inheritance division~\citep{GP15spliddit}.

It is worth noting that while the aforementioned settings involve resources that are considered \emph{desirable} (also known as \emph{goods}), there are many real-world situations where a fair distribution of \emph{undesirable} resources (also known as \emph{chores}) is needed~\citep{G78aha}. Common examples of such situations include the division of household tasks like cooking or cleaning~\citep{IY23kajibuntan}, as well as the distribution of responsibilities for tackling global issues such as climate change among countries~\citep{T02fair}.

The problem of fair division of indivisible chores involves a set of discrete resources for which agents have non-positive values. The goal is to assign each chore to exactly one agent such that the final allocation is fair. A well-studied notion of fairness is \emph{envy-freeness}~\citep{GS58puzzle,F67resource} which requires that each agent weakly prefers its bundle over any other agent's. However, due to the discrete nature of the tasks, an envy-free allocation may not always exist. This has led to the study of approximations such as \emph{envy-freeness up to one chore} (\EF{1}) which bounds the pairwise envy by the removal of some chore in the envious agent's bundle~\citep{B11combinatorial,ACI+19fair}. Unlike exact envy-freeness, an \EF{1} allocation of chores is guaranteed to exist even under general monotone valuations~\citep{LMM+04approximately,BSV21approximate}.

A common assumption in the fair division literature is that any item can be feasibly assigned to any agent. This assumption may not hold in many settings of interest. For example, in course allocation, a student can only attend at most one course at any given time. Similarly, in assigning volunteers to conference sessions, temporal overlaps may need to be taken into account. In such settings, it is more natural to model \emph{conflicts} among the items and allow only feasible (or non-conflicting) allocations.

We formalize the problem of fair and efficient scheduling of indivisible chores under conflict constraints. Each chore is associated with a start time and a finish time. Indivisibility dictates that a chore can be assigned to at most one agent. An agent can perform at most one chore at a time; furthermore,  a chore once started must be performed until its completion. By modeling the chores as vertices of a graph and capturing temporal conflicts with edges, we obtain the problem of dividing the vertices of an \emph{interval} graph among agents such that each agent gets an independent subset. Note that due to conflicts, it may not be possible to allocate all chores. Thus, we ask for schedules to be \emph{maximal}, i.e., it should not be possible to assign any agent an unallocated chore without creating a conflict.
 
\subsection*{Our Contributions}

We initiate the study of fair 
and efficient 
interval scheduling of indivisible chores under conflict constraints and make the following contributions:

\paragraph{Non-existence results:} In \Cref{sec:Non-Existence-Results}, we show that the strongest approximation notion---envy-freeness up to \emph{any} chore (\EFX{})---may not be compatible with maximality. By weakening the fairness requirement to envy-freeness up to one chore (\EF{1}) but strengthening the efficiency requirement to Pareto optimality, we again obtain a non-existence result (see \Cref{fig:relations}). Notably, our negative results hold even for two agents with identical valuations and even when the conflict graph is a path graph. Thus, we focus on \EF{1} and maximality in search of positive results. In \Cref{sec:Limitations}, we highlight the limitations of algorithms from the unconstrained setting---specifically, round robin and envy-cycle elimination---in computing an \EF{1} and maximal schedule.
\paragraph{Results for two agents:} In \Cref{sec:Results_Two_Agents}, we prove our main result: A polynomial-time algorithm for finding an \EF{1} and maximal schedule for two agents under general \emph{monotone} valuations and \emph{any} interval graph. Our analysis develops a novel notion of \emph{adjacent} schedules and uses a \emph{coloring} technique that may be of independent interest.
\paragraph{Results for an arbitrary number of agents:} In \Cref{sec:Results_n_Agents}, we consider the case of an arbitrary number of agents. We show that under \emph{identical} additive valuations, an \EF{1} and maximal schedule always exists for three or more agents for a \emph{path} graph~(\Cref{thm:n_Agents_Identical}). 
Thus, together with the aforementioned result for two agents in \Cref{sec:Results_Two_Agents}, we resolve the existence question for any number of identical agents on a path. The question of finding such schedules in polynomial time remains open (except for the case of four agents). Notably, our proof uses the ``cycle plus $n$-cliques'' theorem, which generalizes a well-known conjecture of Paul Erd\H{o}s. It is the first instance of this theorem being applied to fair division that we are aware of. 

If the valuations are additionally assumed to be dichotomous and there are four or more agents, such a schedule can be found in polynomial time. Finally, we present an efficient algorithm for computing an \EF{1} and maximal schedule for $n$ identical additive agents for graphs where each connected component is of size at most $n$~(\Cref{thm:n_Agents_Identical_Bounded_Components}).

\begin{figure}[t]
	\begin{center}
	    \scalebox{1}{
	       \begin{tikzpicture}[scale=0.9, every node/.style={scale=0.9}]
                \centering
	            \tikzstyle{onlytext}=[]
	            \tikzset{venn circle/.style={circle,minimum width=0mm,fill=#1,opacity=0.1}}
                    \node[onlytext] (Result) at (6.5,-0.75) {{Results}};
                    %
                    \draw [line width=28pt,opacity=0.4,color=red!55,line cap=round,rounded corners] (-0.5,-2) -- (0.5,-2) -- (2.5,-4.5) -- (3.5,-4.5) -- (5,-2) -- (7.75,-2);
                    \draw [line width=28pt,opacity=0.4,color=red!55,line cap=round,rounded corners] (-0.5,-4.5) -- (0.5,-4.5) -- (2.5,-2) -- (3.5,-2) -- (5,-3.25) -- (7.75,-3.25);
                    \draw [line width=28pt,opacity=0.4,color=green!55,line cap=round,rounded corners] (-0.5,-4.5) -- (8.1,-4.5);
	            \node[onlytext] () at (6.5,-2) {\begin{tabular}{c}{Can fail to exist}\\{(\Cref{eg:EFX_Maximal_Counterexample})}\end{tabular}};
                    %
                    %
	            \node[onlytext] () at (6.5,-3.25) {\begin{tabular}{c}{Can fail to exist}\\{(\Cref{eg:EF1_PO_line_graph_counterexample})}\end{tabular}};
                    %
                    %
	            \node[onlytext] () at (6.5,-4.5) {\begin{tabular}{c}{Exists in certain settings}\\{(See \Cref{tab:Summary})}\end{tabular}};
                    \draw[-, line width=1pt] (-1,-1.25) -- (7.5,-1.25);
                    \node[onlytext] (Efficiency) at (3,-0.75) {\begin{tabular}{c}{Efficiency}\\{Notions}\end{tabular}};
                    \node[onlytext] (PO) at (3,-2) {\PO{}};
                    \node[onlytext] (Maximal) at (3,-4.5) {Maximal};
                    \node[onlytext] (Fairness) at (0,-0.75) {\begin{tabular}{c}{Fairnesss}\\{Notions}\end{tabular}};
                    \node[onlytext] (EFX) at (0,-2) {\EFX{}};
                    \node[onlytext] (EFone) at (0,-4.5) {\EF{1}};
                    \draw[->, line width=1pt] (EFX) -- (EFone);
                    \draw[->, line width=1pt] (PO) -- (Maximal);
	       \end{tikzpicture}
            }
	\end{center}
	\caption{Summary of our results. The arrows denote logical implications between fairness and efficiency notions. The positive and negative results are shown in green and red, respectively.}
	\label{fig:relations}
\end{figure}

\begin{table}[t]
\centering
\begin{tabular}{ccc}
No. of agents & Path graph & Interval graph \\
\midrule
\multirow{2}{*}{$n=2$} & \cmark{} for monotone valuations & \cmark{} for monotone valuations \\ 
& (\Cref{{thm:Two_Agents_Line_Graph}}) & (\Cref{thm:Two_Agents_EF1_Interval})\\
\midrule
\multirow{3}{*}{arbitrary $n$} & \cmark{} for identical additive & \cmark{} for identical additive valuations\\
& valuations and $n \geq 3$ & and bounded components\\
& (\Cref{thm:n_Agents_Identical_Dichotomous}) & (\Cref{thm:n_Agents_Identical_Bounded_Components})\\
\end{tabular}
\vspace{0.1in}
\caption{Summary of our results for \EF{1} and maximality. In each cell, a \cmark{} denotes that an \EF{1} and maximal schedule always exists under the assumptions on the number of agents (rows) and the conflict graph (column). With the exception of the case of arbitrary $n$ and path graph, our existence results also provide efficient algorithms. In this case, we get a polynomial-time algorithm for four or more agents with the additional restriction of dichotomous valuations.}
\label{tab:Summary}
\end{table}

\paragraph{Follow-up work.} After the publication of a preliminary version of this work in AAMAS 2024~\citep{EGK+24fair}, a beautiful work of \cite{IMY25dividing} further considered this problem and significantly extended our results. In particular, they show a much more general version of our main result (\cref{thm:Two_Agents_Line_Graph}), extending the theorem from path graphs to any interval or bipartite graph. Additionally, they extend the existence guarantee to \emph{all} graphs, at the cost of changing the running time guarantee from polynomial to pseudo-polynomial. We recommend the reader to refer to their results and proof as it further refines our notion of adjacency (\cref{def:adjacent}) to what they call a ``gapless chain'', maintaining the properties that we require and at the same time, making the argument more intuitive and much simpler.

\subsection{Related Work}
\label{subsec:Related_Work}

Fair allocation of indivisible resources has been predominantly studied in the \emph{unconstrained} setting wherein there are no feasibility constraints on the resources. This problem is a special case of our model when the conflict graph contains no edges. For indivisible \emph{goods}, it is known that an allocation that is envy-free to up one good~(\EF{1}) always exists. Furthermore, the round robin algorithm returns such an allocation in polynomial time under additive valuations, while the envy-cycle elimination algorithm finds an \EF{1} allocation for the larger class of monotone valuations~\citep{LMM+04approximately}. Additionally, under additive valuations, an allocation that is 
\EF{1} and Pareto optimal $(\PO{})$ is known to always exist~\citep{CKM+19unreasonable} and such an allocation can be computed in pseudopolynomial time~\citep{BKV18Finding}. The existence of an allocation satisfying envy-freeness up to any good~(\EFX{}) remains unresolved for four or more agents~\citep{CGM20efx}; though, it is known that \EFX{} is incompatible with Pareto optimality~\citep{PR20almost}.

For indivisible \emph{chores}, an allocation satisfying envy-freeness up to one chore (\EF{1}) is known to exist under monotone valuations~\citep{BSV21approximate}. However, for additive valuations, it is not known if \EF{1} and Pareto optimality can be simultaneously achieved for four or more agents~\citep{GMQ23new} or if envy-freeness up to any chore (\EFX{}) can be satisfied for three or more agents. \footnote{Though, positive results are known for dichotomous valuations~\citep{GMQ22fair,EPS22how,ZW22approximately}.} 
Just as with goods, \EFX{} and Pareto optimality are known to be incompatible for chores~\citep{PR20almost}.

A growing line of research in fair division has studied \emph{feasibility constraints} on the resources~\citep{S21constraints}. For example, several works model the items as vertices of a graph and require each agent's bundle to constitute a connected subgraph~\citep{BCE+17fair,BCL19chore,BCF+22almost}. By contrast, our work requires the bundles to be independent sets.

\citet{HH22fair} study fair allocation of items on a graph when each agent gets an independent subset of the vertices. Their model differs from ours in two ways: First, they require \emph{complete} allocations whereas we focus on maximal allocations. (In our model, a complete allocation may fail to be \EF{1} even for a star graph.) Secondly, the items in their model are goods whereas we consider chores. \citet{CKM+23fair} study a similar model as \citet{HH22fair} but allow for partial allocations. However, they focus on maximizing the egalitarian welfare (i.e., maximizing the minimum utility) as opposed to satisfying approximate envy-freeness. \citet{BKK+23algorithmic} generalize the model of \citet{HH22fair} by incorporating capacity constraints for agents and items.

Assigning independent subsets to agents naturally corresponds to a partial coloring of the conflict graph. A seminal result in this line of work is the \emph{Hajnal-Szemer\'{e}di theorem}~\citep{HS87proof}, which states that any graph with maximum degree $\Delta$ admits an equitable coloring when there are at least $\Delta + 1$ colors.\footnote{An \emph{equitable} coloring is a proper coloring in which the color classes are almost balanced. That is, each vertex is assigned a color such that no pair of adjacent vertices have the same color and any two color classes differ in size by at most one.} The distinguishing point with our work is that this result only focuses on equalizing the \emph{number} of vertices in each color class, whereas our model accounts for the (possibly differing) \emph{valuations} assigned by the agents to the vertices.

Our work is closely related to that of~\citet{LLZ21fair}, who studied fair scheduling in the context of indivisible \emph{goods} (as opposed to chores). Their framework is somewhat more general than ours because they allow for \emph{flexible} intervals; that is, the processing time of a job can be strictly less than its finish time minus its start time. They study approximate maximin fair share (\MMS{}) and \EF{1} notions. Importantly, they show that a schedule that maximizes Nash social welfare (i.e., the geometric mean of agents' utilities) satisfies $\nicefrac{1}{4}$-\EF{1} and Pareto optimality. For chores, however, maximizing or minimizing the product of (absolute value of) agents' utilities can violate \EF{1} or Pareto optimality.

In the literature on scheduling problems~\citep{L04handbook}, various other fairness criteria have been studied such as minimizing the maximum deviation from a desired load~\citep{AAN+98fairness}, minimizing the $\ell_p$ norm of flow times~\citep{IM20fair}, and analyzing the welfare degradation due to imposition of fairness constraints~\citep{BFF+16price}. These notions differ from the ``up to one item'' style approximations studied in our work.

\section{Preliminaries}
\label{sec:Preliminaries}

Given any $r \in \mathbb{N}$, let $[r] \coloneqq \{1,2,\dots,r\}$. 

\paragraph{Problem instance.}
An instance of the 
\emph{chore interval scheduling problem} (\CISP{})
is given by a tuple $\langle \A, \C, \T, \V \rangle$, where $\A \coloneqq \{a_1,\dots,a_n\}$ is the set of $n$ \emph{agents} (or machines), $\C \coloneqq \{c_1,\dots,c_m\}$ is the set of $m$ indivisible \emph{chores} (or jobs), $\T$ is the set containing the \emph{timing information} of all chores, and $\V$ is the \emph{valuation profile}. The sets $\T$ and $\V$ are formally defined below.

\paragraph{Timing information.}
Each chore $c_j \in \C$ is associated with a \emph{start time} $s_j \in \mathbb{N} \cup \{0\}$ 
and a \emph{finish time} $f_j \in \mathbb{N}$ with $f_j > s_j$. 
The set $\T$ contains the tuple $(s_j,f_j)$ for every chore $c_j \in \C$, i.e., $\T \coloneqq \{ (s_j,f_j)_{c_j \in \C} \}$. 
Two chores $c_i, c_j \in \C$ overlap if and only if their intervals have a non-empty intersection i.e., $s_i < f_j $ and $ s_j < f_i$. 

\paragraph{Feasibility.}
An agent can perform at most one chore at any time instant. 
A chore $c_j$ is performed successfully if it is performed 
by an agent during the interval $[s_j,f_j)$. A set of chores $C \subseteq \C$ is said to be \emph{feasible} for an agent $a_i$ if all chores in $C$ can be performed successfully by $a_i$ without overlap. We will write $\F$ to denote the \emph{feasibility set} of any agent, which is the set of all feasible subsets of chores, i.e.,
\begin{equation}
\label{eq:F}
\F \coloneqq \left\{ C \subseteq \C \mid \forall c_j \neq c_{ j' } \in C : [ s_j, f_j ) \cap [ s_{ j' }, f_{ j' } ) = \emptyset \right\} .
\end{equation}
Observe that any subset of pairwise non-overlapping chores is feasible for any agent; thus, the feasibility set is the same for all agents.

\paragraph{Valuation functions.}
The valuation function $v_i : \F \rightarrow \mathbb{Z}_{\leq 0}$ specifies the value (or utility) derived by agent $a_i$ for every feasible set of chores. Notice that the valuations are non-positive integers. A valuation function $v_i$ is \emph{monotone} if for any two feasible subsets of chores $C,C' \in \F$ such that $C \subseteq C'$, we have $v_i(C) \geq v_i(C')$. We say that the valuations are \emph{additive} if for any feasible set of chores $C \in \F$, $v_i(C) \coloneqq \sum_{c_j \in C} v_i(\{c_j\})$. For simplicity, we will write $v_{i,j}$ to denote $v_i(\{c_j\})$. The valuation profile $\V \coloneqq \{v_1,v_2,\dots,v_n\}$ specifies the valuation functions of all agents. We say that agents have \emph{identical} valuations if the valuation functions $v_1(\cdot),v_2(\cdot),\dots,v_n(\cdot)$ are the same.

\paragraph{Schedule.}
A \emph{schedule} (or \emph{allocation}) $\X \coloneqq (X_1,X_2,\dots,X_n)$ is a sequence of pairwise disjoint subsets of~$\C$, where $X_i$ denotes the set of chores (or \emph{bundle}) assigned to agent $a_i$. Thus, for every pair of agents $a_i,a_k \in \A$, we have $X_i \cap X_k = \emptyset$ and $X_1 \cup \dots \cup X_n \subseteq \C$. A schedule $\X$ is said to be \emph{feasible} if for every agent $a_i$, we have $X_i \in \F$. 

\paragraph{Complete and maximal schedule.} A schedule $\X \coloneqq (X_1,X_2,\dots,X_n)$ is called \emph{complete} if no chore is left unassigned, i.e., $X_1 \cup \dots \cup X_n = \C$. Note that a complete schedule may not be feasible in general; indeed, for two agents and three chores that have identical start times and finish times, no complete schedule is feasible. A \emph{maximal} schedule is one that is feasible and has the additional property that assigning any unallocated chore to any agent makes it infeasible. Unless explicitly stated otherwise, we will use the term `schedule' to refer to a feasible (but possibly incomplete and possibly non-maximal) schedule.

\paragraph{Conflict graph.}
Given any \CISP{} instance $\langle \A, \C, \T, \V \rangle$, we will find it convenient to define its \emph{conflict graph} $G = (\C,E)$ as follows~\citep{HH22fair,CKM+23fair}: The set of vertices is the set of chores, and for any pair of vertices $c_i,c_j \in \C$, there is an undirected edge $\{c_i,c_j\}$ if and only if $c_i$ and $c_j$ overlap; see \Cref{fig:Conflict_graph_example} for an illustration. 
\begin{figure}[t]
\centering
\begin{subfigure}[b]{0.45\linewidth}
	\centering
	\begin{tikzpicture}[line width=0.8pt]
		\tikzmath{\choreHeight=0.4;}
		\tikzmath{\gap=0.75;}
			\draw[help lines, dashed, step = 1cm] (0, -0.3) grid[ystep=0] (3, 1.5);
			\draw (0,0) rectangle node{$c_1$} (3,\choreHeight);
			\draw (0,\gap) rectangle node{$c_2$} (1,\gap+\choreHeight);
			\draw (1,\gap) rectangle node{$c_3$} (2,\gap+\choreHeight);
			\draw (2,\gap) rectangle node{$c_4$} (3,\gap+\choreHeight);
			\node[] at (0,-0.4){\scriptsize{$t=0$}};
			\node[] at (3,-0.4){\scriptsize{$t=3$}};		
	\end{tikzpicture}
	\caption{}
\end{subfigure}
\begin{subfigure}[b]{0.45\linewidth}
	\centering
	\begin{tikzpicture}
		\tikzset{mynode/.style = {shape=circle,draw,inner sep=1pt}} 
		\node[mynode] (1) at (1,1) {$c_1$};
		\node[mynode] (2) at (0,0) {$c_2$};
		\node[mynode] (3) at (1,0) {$c_3$};
		\node[mynode] (4) at (2,0) {$c_4$};
		\draw (1) -- (2);
		\draw (1) -- (3);
		\draw (1) -- (4);
	\end{tikzpicture} 
	\caption{}
\end{subfigure}
\caption{(a) A scheduling instance and (b) its conflict graph.}
\label{fig:Conflict_graph_example}
\end{figure}
Note that a maximal schedule corresponds to a subpartition into independent sets in the conflict graph. Also, observe that the conflict graphs correspond to the class of \emph{interval} graphs~\citep[Sec 5.1.15]{W01introduction}.

\paragraph{Envy-freeness.} A schedule $\X \coloneqq (X_1,X_2,\dots,X_n)$ is said to be (a) \emph{envy-free} (\EF{})~\citep{GS58puzzle,F67resource} if for every pair of agents $a_i,a_k \in \A$, we have $v_i(X_i) \geq v_i(X_k)$; (b) \emph{envy-free up to any chore} (\EFX{})~\citep{CKM+19unreasonable,ACI+19fair} if for every pair of agents $a_i,a_k \in \A$ 
and for every chore $c \in X_i$, we have $v_i(X_i \setminus \{c\}) \geq v_i(X_k)$, and (c) \emph{envy-free up to one chore} (\EF{1})~\citep{B11combinatorial,ACI+19fair} if for every pair of agents $a_i,a_k \in \A$ such that $X_i \neq \emptyset$, we have $v_i(X_i \setminus \{c\}) \geq v_i(X_k)$ for some chore $c \in X_i$. It is easy to verify that $\EF{} \Rightarrow \EFX{} \Rightarrow \EF{1}$ and that all implications are strict.

\paragraph{Pareto optimality.}
A schedule $\X \coloneqq (X_1,X_2,\dots,X_n)$ is said to Pareto dominate another schedule $\Y \coloneqq \{Y_1,Y_2,\dots,Y_n\}$ if $v_i(X_i) \geq v_i(Y_i)$ for every agent $a_i$ and $v_k(X_k) > v_k(Y_k)$ for some agent $a_k$. A \emph{Pareto optimal} schedule is one that is maximal and is not Pareto dominated by any other maximal schedule. The maximality requirement is helpful in ruling out trivial solutions; in particular, an empty schedule that leaves all chores unassigned cannot be Pareto optimal according to this definition.

\section{Non-Existence Results}
\label{sec:Non-Existence-Results}

In this section, we will show that various combinations of fairness and efficiency notions can fail to exist in the chore interval scheduling problem. Interestingly, all of our counterexamples involve two agents with identical valuations and the conflict graph is a path graph.

Let us start with a negative result for \EFX{} and maximality.

\begin{example}[\EFX{} and maximal schedule may not exist]
Consider an instance with two agents $a_1$ and $a_2$ and four chores $c_1$, $c_2$, $c_3$, and $c_4$ that are identically valued by the agents at $-1$, $-1$, $-1$, and $-4$, respectively. The conflict graph is shown below.

\begin{figure}[h]
\centering
\begin{tikzpicture}
    \tikzset{mynode/.style = {shape=circle,draw,inner sep=1pt}} 
    \node[mynode] (1) at (0,0) {$c_1$};
    \node[mynode] (2) at (1,0) {$c_2$};
    \node[mynode] (3) at (2,0) {$c_3$};
    \node[mynode] (4) at (3,0) {$c_4$};
    \draw (1) -- (2);
    \draw (2) -- (3);
    \draw (3) -- (4);
\end{tikzpicture}
\label{fig:EFX_maximal_counterexample}
\end{figure}
Let $\X$ be the desired \EFX{} and maximal schedule. Observe that due to maximality, the chore $c_4$ cannot remain unallocated under $\X$. This is because if the neighboring chore $c_3$ is assigned to one of the agents, say $a_1$, then the chore $c_4$ must be assigned to the other agent $a_2$. Similarly, if $c_3$ is unassigned, then $c_4$ can be assigned to either of the agents without creating any conflict. Thus, we can assume, without loss of generality, that $c_4$ is assigned to agent $a_1$.

In order for the schedule $\X$ to satisfy \EFX{}, agent $a_1$ cannot be assigned any other chore. Furthermore, feasibility dictates that the other agent $a_2$ can be given at most two of the three remaining chores $c_1$, $c_2$, and $c_3$. If agent $a_2$ gets exactly two chores, then it must be given $c_1$ and $c_3$; however, then $c_2$ must be assigned to agent $a_1$, violating \EFX{}. On the other hand, if agent $a_2$ gets at most one chore, then once again by maximality of $\X$, agent $a_1$ will be required to get at least one chore out of $c_1$, $c_2$, and $c_3$, again violating \EFX{}. Thus, an \EFX{} and maximal schedule does not exist in the above instance.\qed
\label{eg:EFX_Maximal_Counterexample}
\end{example}

The non-existence of \EFX{} motivates the consideration of a weaker approximation such as \EF{1}. Our next example shows that \EF{1} and Pareto optimality can be mutually incompatible even for two agents with identical valuations on a path graph.

\begin{example}[\EF{1} and Pareto optimal schedule may not exist]
Consider an instance with two agents $a_1$, $a_2$ and five chores $c_1, \dots, c_5$ that are identically valued by the agents at $-2$, $-10$, $-1$, $-10$, and $-2$, respectively.  The conflict graph is shown below.

\begin{figure}[h]
\centering
\begin{tikzpicture}
    \tikzset{mynode/.style = {shape=circle,draw,inner sep=1pt}} 
    \node[mynode] (1) at (0,0) {$c_1$};
    \node[mynode] (2) at (1,0) {$c_2$};
    \node[mynode] (3) at (2,0) {$c_3$};
    \node[mynode] (4) at (3,0) {$c_4$};
    \node[mynode] (5) at (4,0) {$c_5$};
    \draw (1) -- (2);
    \draw (2) -- (3);
    \draw (3) -- (4);
    \draw (4) -- (5);
\end{tikzpicture}
\label{fig:EF1_PO_line_graph_counterexample}
\end{figure}

Any maximal schedule that allocates a ``heavy'' chore ($c_2$ or $c_4$) to one of the agents, say $a_1$, can be shown to be Pareto dominated by a schedule that assigns the extreme chores ($c_1$ and $c_5$) to $a_1$ and the middle chore $c_3$ to $a_2$. Therefore, any maximal schedule that leaves both heavy chores unassigned must be of the form $(\{c_1,,c_5\},\{c_3\})$ or $(\{c_3\},\{c_1,,c_5\})$, neither of which satisfy \EF{1}.\qed
\label{eg:EF1_PO_line_graph_counterexample}
\end{example}

Note that a Pareto optimal schedule (without \EF{1}) always exists; indeed, a schedule that maximizes the sum of agents' utilities~(i.e., the utilitarian social welfare) among all maximal schedules is Pareto optimal.

Another notion of efficiency is \emph{completeness} which asks that no chore should be left unassigned. A complete schedule may not exist (consider two agents and a triangle conflict graph). However, even when a complete schedule exists, no such schedule may satisfy \EF{1}.

\begin{example}[\EF{1} and complete schedule may not exist]
Consider an instance with two agents 
and four chores $c_1,c_2,c_3,c_4$ such that the odd index chores are valued at $-1$ each and the even index chores are valued at $-3$ each by both agents. The conflict graph is a path graph, similar to the one in \Cref{eg:EFX_Maximal_Counterexample}. 
Any complete schedule must assign all odd index chores to one agent and all even index chores to the other. It is easy to see that \EF{1} is violated from the perspective of the agent with even index chores.\qed
\label{eg:EF1_complete_counterexample}
\end{example}

The failure of \EF{1} and completeness means that we must focus on \EF{1} and maximality in search of positive results. 
Note that for three or more agents and a path graph, a schedule is maximal if and only if it is complete.

\section{Limitations of Algorithms from the Unconstrained Setting}
\label{sec:Limitations}

One of the reasons why the chore interval scheduling problem is difficult is that algorithms for fair division without constraints do not necessarily work for the constrained problem. We will illustrate this challenge through two examples that demonstrate that the well-known round robin and envy-cycle elimination algorithms, which satisfy \EF{1} for unconstrained items, fail to do so in the presence of conflicts.

The round robin algorithm iterates over the agents according to a fixed permutation, and each agent, on its turn, picks its favorite remaining item. This algorithm is known to find an \EF{1} allocation under additive valuations in the unconstrained problem. However, when the conflict graph is a path graph, round robin can fail to achieve \EF{1} even for two agents with identical valuations.

\begin{example}[Round robin fails \EF{1}]
Consider an instance with six chores $c_1,\dots,c_6$ and two agents $a_1,a_2$ with identical valuations. The valuations profile and the scheduling constraints are shown in \Cref{{fig:RoundRobin_Fails_EF1}}.

\begin{figure}[ht]
\centering
\begin{subfigure}[b]{\linewidth}
	\centering
	\begin{tikzpicture}[line width=0.8pt]
	\tikzmath{\choreHeight=0.4;}
	\tikzmath{\gap=0.75;}
 		\draw[help lines, dashed, step = 0.5cm] (0, -0.3) grid[ystep=0] (3.5, 1.5);
		\draw (0,0) rectangle node{$c_1$} (1,\choreHeight);
		\draw (0.5,\gap) rectangle node{$c_2$} (1.5,\gap+\choreHeight);
		\draw (1,0) rectangle node{$c_3$} (2,0+\choreHeight);
		\draw (1.5,\gap) rectangle node{$c_4$} (2.5,\gap+\choreHeight);
		\draw (2,0) rectangle node{$c_5$} (3,0+\choreHeight);
		\draw (2.5,\gap) rectangle node{$c_6$} (3.5,\gap+\choreHeight);
		\node[] at (0,-0.4){\scriptsize{$t=0$}};
		\node[] at (3.5,-0.4){\scriptsize{$t=7$}};
	\end{tikzpicture}
\end{subfigure}
\begin{subfigure}[b]{\linewidth}
	\centering
	\begin{tikzpicture}[scale=0.8]
		\tikzset{mynode/.style = {shape=circle,draw,inner sep=1pt}} 
            \node[] (0) at (0,2.7) {};
		\node[mynode] (1) at (0,2) {$c_1$};
		\node[mynode] (2) at (1,2) {$c_2$};
		\node[mynode] (3) at (2,2) {$c_3$};
		\node[mynode] (4) at (3,2) {$c_4$};
		\node[mynode] (5) at (4,2) {$c_5$};
		\node[mynode] (6) at (5,2) {$c_6$};
		\draw (1) -- (2);
		\draw (2) -- (3);
		\draw (3) -- (4);
		\draw (4) -- (5);
		\draw (5) -- (6);
		%
		\node[] (1) at (-0.75,1.25) {$a_1:$};
		\node[] (1) at (0,1.25) {0};
		\node[] (2) at (1,1.25) {-5};
		\node[] (3) at (2,1.25) {-2};
		\node[] (4) at (3,1.25) {-1};
		\node[] (5) at (4,1.25) {-3};
		\node[] (6) at (5,1.25) {-6};
		%
		\node[] (1) at (-0.75,0.5) {$a_2:$};
		\node[] (1) at (0,0.5) {0};
		\node[] (2) at (1,0.5) {-5};
		\node[] (3) at (2,0.5) {-2};
		\node[] (4) at (3,0.5) {-1};
		\node[] (5) at (4,0.5) {-3};
		\node[] (6) at (5,0.5) {-6};
	\end{tikzpicture} 
\end{subfigure}
\caption{The instance used in \Cref{eg:RoundRobin_Fails_EF1} (top) and its conflict graph and valuations (bottom).}
\label{fig:RoundRobin_Fails_EF1}
\end{figure}
Suppose agent $a_1$ goes first in the round robin algorithm. Then the induced schedule is $X_1 = \{c_1,c_3,c_5\}$ and $X_2 = \{c_4,c_2,c_6\}$. Note that $v_2(X_2 \setminus c) < v_2(X_1)$ for every $c \in X_2$, implying that the schedule fails \EF{1}. 

The reason why round robin fails \EF{1} in the scheduling model is that after picking $c_4$ in the first round, agent $a_2$ can no longer pick $c_5$ in its next turn due to feasibility constraint. Thus, it has to pick its favorite among the feasible chores, which is $c_2$. This results in a ``build up'' of envy in each subsequent round, which cannot be compensated even after removing the worst chore $c_6$ in its bundle.\qed
\label{eg:RoundRobin_Fails_EF1}
\end{example}

The envy-cycle elimination algorithm is known to find an \EF{1} allocation for indivisible goods in the unconstrained setting. For chores, a modification of this algorithm called the top-trading envy-cycle elimination algorithm~\citep{BSV21approximate} provides similar guarantees. The latter algorithm works as follows: At each step, the algorithm assigns a chore to a ``sink'' vertex in the envy graph.\footnote{The envy graph associated with a given schedule is a directed graph whose vertices are the agents and there is an edge $(i,j)$ if agent $i$ envies agent $j$ in the given schedule. A sink vertex refers to an agent who does not envy any other agent.} If a sink does not exist, there must exist a cycle of ``most envied edges'', which, when resolved~(i.e., by a cyclic swap of bundles), creates a sink vertex.

\begin{example}[Envy-cycle elimination fails \EF{1}]
Consider an instance with two agents $a_1,a_2$ and five chores $c_1,\dots,c_5$ whose conflict graph is a path graph (similar to \Cref{eg:EFX_Maximal_Counterexample}). The chores are valued at $-10$, $-1$, $-10$, $-3$, and $-2$, respectively, by both agents. Note that since the valuations are identical, an envy cycle can never occur. Thus, a sink agent always exists.

At each step of the algorithm, we allow a sink agent to pick its favorite among the feasible chores. Thus, agent $a_1$ starts by picking the chore $c_2$, which makes $a_2$ the sink. Next, agent $a_2$ picks the chore $c_5$, which again makes $a_1$ the sink, and so on. It can be checked that the induced schedule is $(\{c_2,c_4\},\{c_1,c_3,c_5\})$. However, \EF{1} is violated from $a_2$'s perspective as it gets both heavy chores.\qed
\label{eg:Envy-cycle-fails-EF1}
\end{example}

Another challenge with the chore interval scheduling problem is that the set of maximal schedules seems to lack any ``easy to exploit'' structure. In particular, one might wonder whether maximal schedules constitute a \emph{matroid}~\citep{O22matroid}. 
    %
    %
If so, it would be natural to use known algorithms for finding an \EF{1} allocation under matroid constraints~\citep{BB19matroid,DFS23fair}. 
However, any maximal schedule lacks the downward closedness property, since unassigning an item from a maximal schedule does not maintain maximality. Similarly, the exchange axiom may also fail for a maximal schedule. To see this, consider a star graph as shown in \Cref{fig:Conflict_graph_example}. Suppose the center chore is assigned to agent $a_1$ and all leaf chores are assigned to agent $a_2$. Agent $a_1$'s bundle has fewer chores, but transferring any of agent $a_2$'s chores to agent $a_1$ creates infeasibility. Thus, it is not clear if ideas from matroid theory can be applied to the chore interval scheduling problem.

\section{Results for Two Agents}
\label{sec:Results_Two_Agents}

In this section, we will discuss our algorithmic results for two agents. We will start with the case when the conflict graph is a \emph{path} graph (\Cref{thm:Two_Agents_Line_Graph}). As we have seen in \Cref{sec:Non-Existence-Results,sec:Limitations}, the setting of two agents and a path graph is already quite nontrivial as all of our counterexamples hold even in this special case. Our algorithm will use a novel \emph{coloring} technique and a construct called \emph{adjacent schedules} that will help us circumvent the limitations discussed previously. We will then show that, building on our coloring technique, a more sophisticated algorithm can find an \EF{1} and maximal schedule for two agents for any \emph{interval} graph~(\Cref{thm:Two_Agents_EF1_Interval}).

Notably, 
our algorithms apply to general \emph{monotone} valuations, which is a significantly broader class that contains additive valuations. In light of the non-existence results in \Cref{sec:Non-Existence-Results}, the result in \Cref{thm:Two_Agents_EF1_Interval} is the most general positive result for two agents that one can expect in our problem.

Let us start with our algorithm for path graphs.

\begin{restatable}[Two agents and path graph]{theorem}{TwoAgentsLineGraph}
There is an algorithm that, given any \CISP{} instance with two agents with monotone valuations and an arbitrary path graph, returns an \EF{1} and maximal schedule. Furthermore, this algorithm runs in polynomial time, assuming access to a value query oracle.
\label{thm:Two_Agents_Line_Graph}
\end{restatable}

Our algorithm will use the idea of \emph{adjacent schedules} which is defined below.

\begin{dfn}[Adjacent schedules]
\label{def:adjacent}
Two schedules $\X = (X_1,X_2)$ and $\Y = (Y_1,Y_2)$ are said to be adjacent if for any $i \in \{1,2\}$, $Y_i$ is obtained from $X_i$ by adding at most one element and removing at most one element, that is, 
\[\abs{Y_i\setminus X_i} \leq 1 \; and \; \abs{X_i\setminus Y_i} \leq 1 \; \text{for } i \in \{1,2\}.\]
\end{dfn}
For example, an adjacent schedule can be obtained by the two agents exchanging a pair of chores, or by assigning an unallocated chore to agent 1 and transferring another chore from agent 1 to agent 2, etc. 
The following lemma shows that if we have two adjacent schedules and an agent is envious in one but not in the other, then one of the two schedules or the schedules obtained after swapping the bundles satisfies \EF{1}.

\begin{restatable}[]{lemma}{AdjacentScheduleLemma}
Let $\X=(X_1,X_2)$ and $\Y=(Y_1,Y_2)$ be two adjacent schedules such that
\begin{itemize}
\item in schedule $\X$, agent $a_1$ envies agent $a_2$, and
\item in schedule $\Y$, agent $a_1$ does not envy agent $a_2$. 
\end{itemize}
Then, at least one of the following four schedules must be \EF{1}: $\X$, $\Y$, or  the  schedules obtained by swapping the agents' bundles, i.e., $\X':=(X_2,X_1)$, or $\Y':=(Y_2,Y_1)$.  
\label{lem:Adjacent_EF1}
\end{restatable}
\begin{proof} (of \Cref{lem:Adjacent_EF1})
For the sake of contradiction, assume that none of the four schedules is \EF{1}. 
We can assume that $a_2$ does not envy $a_1$ in schedule $\X$, because otherwise $\X'$ would be \EF{}. Similarly, we can assume that $a_2$ does not envy $a_1$ in schedule $\Y'$, because otherwise $a_2$ would not envy $a_1$ in $\Y$ and $\Y$ would be \EF{}. Now, since we assumed $\X$ and $\Y'$ are not \EF{1}, in both schedules agent $a_1$ must envy agent $a_2$ even after the giving up any single chore. By adjacent property of $\X$ and $\Y$, we know that $|X_1\setminus Y_1| \leq 1$ and $|Y_2 \setminus X_2|\leq 1$.
Hence, we can write
\begin{align*}
    \text{For } \X: v_1(X_1 \setminus (X_1\setminus Y_1 )) &< v_1(X_2), \text{ and} \\
    \text{For } \Y': v_1(Y_2 \setminus (Y_2 \setminus X_2) )  &< v_1(Y_1). 
\end{align*}
Note  that $X_1 \setminus \{X_1 \setminus Y_1\} \subseteq Y_1$ and $Y_2 \setminus \{Y_2 \setminus X_2\} \subseteq X_2$.
Putting this together with the above two equations, we get
\begin{align*}
   v_1(Y_1) \leq v_1(X_1 \setminus \{X_1\setminus Y_1\}) &< v_1(X_2) \text{ and} \\
   v_1(X_2) \leq v_1(Y_2 \setminus \{Y_2 \setminus X_2\})  &< v_1(Y_1). 
\end{align*}
The two inequalities give us a contradiction.
\end{proof}
Note that while the conditions of the lemma only require the envy to be considered from agent $a_1$'s perspective, the \EF{1} implication holds for both agents. Also, 
the proof of \Cref{lem:Adjacent_EF1} works for monotone valuations. Furthermore, the lemma can also be shown to hold for \emph{goods} (i.e., when all valuations are non-negative) under monotone valuations.

\Cref{lem:Adjacent_EF1} implies that in order to find an \EF{1} and maximal schedule, it suffices to find a pair of adjacent and maximal schedules such that agent $a_1$ is envious under one but not the other. Towards this goal, we will construct a \emph{sequence} of maximal schedules, starting with a schedule $\X=(X_1,X_2)$ and ending with the swapped bundles $\X'=(X_2,X_1)$, such that any two consecutive schedules in the sequence are adjacent.
Note that agent $a_1$ cannot envy agent $a_2$ in both $X$ and $X'$. Hence, there will be two consecutive schedules (i.e., a switchover point) in this sequence that will satisfy the condition required in \Cref{lem:Adjacent_EF1} and will give us the desired \EF{1} schedule.

For our next result (\Cref{lem:Adjacent_sequence_Line_Graph}), we will find it convenient to identify agent $a_1$ with {\color{red}red} and agent $a_2$ with {\color{blue}blue} color.

\begin{lemma}
\label{lemma:schedulesequence:path}
For any \CISP{} instance with two agents and a path graph, there exists a sequence of feasible schedules $( \X_1=(R_1,B_1), \X_2=(R_2,B_2), \dots, \X_m=(R_m,B_m))$, where $m \in \mathbb{N}$ is the number of chores, such that 
\begin{enumerate}
\item The last schedule is obtained by swapping
the two bundles in the first schedule. That is, $R_1=B_m$ and $B_1=R_m$.
\item For any $2\leq i \leq m$, the schedules $\X_i$ and $\X_{i-1}$ are adjacent.
\item For any $1 \leq i \leq m$, the schedule $\X_i$ is maximal.
\end{enumerate}
\label{lem:Adjacent_sequence_Line_Graph}
\end{lemma}
Note that \Cref{thm:Two_Agents_Line_Graph} follows readily from \Cref{lem:Adjacent_EF1,lem:Adjacent_sequence_Line_Graph}. Indeed, the sequence of schedules returned by the algorithm in \Cref{lem:Adjacent_sequence_Line_Graph} consists only of maximal schedules. The extreme schedules are swapped versions of each other, therefore there must exist a pair of consecutive schedules, say $\X_i$ and $\X_{i+1}$, in the sequence where the envy of agent $a_1$ switches. From \Cref{lem:Adjacent_EF1}, at least one of $\X_i$, $\X_{i+1}$, or the swapped versions of these must satisfy \EF{1}.

\begin{figure}[t]
\centering
    \begin{tikzpicture}[scale=0.8]
        \tikzset{mynode/.style = {shape=circle,draw,inner sep=1pt}} 
        \node[mynode] (1) at (0,2) {$c_1$};
        \node[mynode] (2) at (1,2) {$c_2$};
        \node[mynode] (3) at (2,2) {$c_3$};
        \node[mynode] (4) at (3,2) {$c_4$};
        \node[mynode] (5) at (4,2) {$c_5$};
        \node[mynode] (6) at (5,2) {$c_6$};
        \draw (1) -- (2);
        \draw (2) -- (3);
        \draw (3) -- (4);
        \draw (4) -- (5);
        \draw (5) -- (6);
        %
        \node[] (0) at (-0.75,1.25) {$\X_1:$};
        \node[] (1) at (0,1.25) {$\color{red}R$};
        \node[] (2) at (1,1.25) {$\color{blue}B$};
        \node[] (3) at (2,1.25) {$\color{red}R$};
        \node[] (4) at (3,1.25) {$\color{blue}B$};
        \node[] (5) at (4,1.25) {$\color{red}R$};
        \node[] (6) at (5,1.25) {$\color{blue}B$};
        %
        \node[] (0) at (-0.75,0.5) {$\X_2:$};
        \node[] (1) at (0,0.5) {$\color{blue}B$};
        \node[] (2) at (1,0.5) {$-$};
        \node[] (3) at (2,0.5) {$\color{red}R$};
        \node[] (4) at (3,0.5) {$\color{blue}B$};
        \node[] (5) at (4,0.5) {$\color{red}R$};
        \node[] (6) at (5,0.5) {$\color{blue}B$};
        \draw[draw=black] (-0.25,0.25) rectangle ++(1.5,0.5);
        \node[] (0) at (-0.75,-0.25) {$\X_3:$};
        \node[] (1) at (0,-0.25) {$\color{blue}B$};
        \node[] (2) at (1,-0.25) {$\color{red}R$};
        \node[] (3) at (2,-0.25) {$-$};
        \node[] (4) at (3,-0.25) {$\color{blue}B$};
        \node[] (5) at (4,-0.25) {$\color{red}R$};
        \node[] (6) at (5,-0.25) {$\color{blue}B$};
        \draw[draw=black] (0.75,-0.5) rectangle ++(1.5,0.5);
        \node[] (0) at (-0.75,-1) {$\X_4:$};
        \node[] (1) at (0,-1) {$\color{blue}B$};
        \node[] (2) at (1,-1) {$\color{red}R$};
        \node[] (3) at (2,-1) {$\color{blue}B$};
        \node[] (4) at (3,-1) {$-$};
        \node[] (5) at (4,-1) {$\color{red}R$};
        \node[] (6) at (5,-1) {$\color{blue}B$};
        \draw[draw=black] (1.75,-1.25) rectangle ++(1.5,0.5);
        \node[] (0) at (-0.75,-1.75) {$\X_5:$};
        \node[] (1) at (0,-1.75) {$\color{blue}B$};
        \node[] (2) at (1,-1.75) {$\color{red}R$};
        \node[] (3) at (2,-1.75) {$\color{blue}B$};
        \node[] (4) at (3,-1.75) {$\color{red}R$};
        \node[] (5) at (4,-1.75) {$-$};
        \node[] (6) at (5,-1.75) {$\color{blue}B$};
        \draw[draw=black] (2.75,-2) rectangle ++(1.5,0.5);
        \node[] (0) at (-0.75,-2.5) {$\X_6:$};
        \node[] (1) at (0,-2.5) {$\color{blue}B$};
        \node[] (2) at (1,-2.5) {$\color{red}R$};
        \node[] (3) at (2,-2.5) {$\color{blue}B$};
        \node[] (4) at (3,-2.5) {$\color{red}R$};
        \node[] (5) at (4,-2.5) {$\color{blue}B$};
        \node[] (6) at (5,-2.5) {$\color{red}R$};
        \draw[draw=black] (3.75,-2.75) rectangle ++(1.5,0.5);
    \end{tikzpicture}
\caption{Illustrating the coloring technique on a path.}
\vspace{-0.1in}
\label{fig:Coloring_Example}
\end{figure}

\begin{proof} (of \Cref{lem:Adjacent_sequence_Line_Graph})
We will prove the lemma by means of a coloring technique. Let the chores, in the increasing order of their finish times, be denoted by $c_{1}, c_{2}, \dots, c_{m}$.
We define the $m$ schedules $\X_i = (R_i,B_i)$ for $1\leq i \leq m$, as follows~(see \Cref{fig:Coloring_Example}):
\begin{itemize}
\item   $R_1 =\{c_{h} : h \text{ is odd} \}$ and $B_1 = \{c_{h} : h \text{ is even}\}$.
\item For any  $i \in \{2,3, \dots, m-1\}$, 
 \begin{itemize}
\item $R_i$ contains
$\{c_{h} : 1 \leq h < i, h \text{ is even} \} \cup 
\{c_{h} : i < h \leq m, h \text{ is odd} \}$,
\item $B_i$ contains
$\{c_{h} : 1 \leq h < i, h \text{ is odd} \} \cup 
\{c_{h} : i < h \leq m, h \text{ is even} \}$,
\item and the chore $c_{i}$ is unassigned.
 \end{itemize}
 \item $R_m =\{c_{h} : h \text{ is even} \}$ and $B_m = \{c_{h} : h \text{ is odd}\}$.
\end{itemize}

\emph{Feasibility:} It is easy to see that each schedule $\X_i$ in the above construction is feasible. Indeed, any colored chore $c_{h}$ intersects with at most one colored chore with earlier finish time, which can only be $c_{{h-1}}$. But the two chores $c_{h}$ and $c_{{h-1}}$ are not assigned to the same agent in any of the schedules.

\smallskip
\emph{Proof of part (1)}: 
From the construction, it is clear that $R_1=B_m$ and $B_1=R_m$. 

\smallskip
\emph{Proof of part (2)}:
For any even $i \in \{2,\dots,m\}$, observe that 
$B_i \setminus B_{i-1} = \{c_{i-1}\}$, and
$B_{i-1} \setminus B_{i} = \{c_{i}\}$.
Moreover, $R_{i-1} \setminus R_{i} \subseteq \{c_1\}$ and $R_{i} \setminus R_{i-1} \subseteq \{c_m\}$. Hence, $\X_{i-1}$ and $\X_{i}$ are adjacent. For odd $i$, a similar argument works.

\smallskip
\emph{Proof of part (3)}:
For any schedule $\X_i$, there is at most one unassigned chore. Any such chore is sandwiched between two chores that are assigned to different agents, implying maximality.
\end{proof}

Having shown that an \EF{1} and maximal schedule exists for two agents on a path graph, let us now turn to the more general class of interval graphs. 
Here, we will first show that a generalization of the coloring technique for path graphs gives a weaker fairness guarantee of \EF{2} (instead of \EF{1}) along with maximality for interval graphs.\footnote{A schedule $\X = (X_1,X_2)$ satisfies envy-freeness up to two chores (\EF{2}) if for every pair of agents $a_i,a_k$, there exist two chores $c,c' \in X_i$ such that $v_i(X_i \setminus \{c,c'\}) \geq v_i(X_k)$.} 

\subsection{EF2 for two agents and interval graph}

Recall from the proof of \Cref{thm:Two_Agents_Line_Graph} 
that by constructing a sequence of adjacent and maximal schedules, we were able to show that an \EF{1} and maximal schedule exists for two agents when the conflict graph is a path. In this section, we will extend this insight to interval graphs. Specifically, for an arbitrary interval conflict graph, we will construct a sequence of adjacent schedules that are \emph{almost maximal} (i.e., can be made maximal by adding at most one chore). By a similar argument as in the proof of \Cref{thm:Two_Agents_Line_Graph}, we will show that an \EF{1} and almost maximal schedule always exists. By adding an extra chore to make the schedule maximal, we obtain a weaker fairness guarantee of \emph{envy-freeness up to two chores} (\EF{2}) in conjunction with maximality. A schedule is said to be envy-free up to two chores (\EF{2}) if for every pair of agents $a_i,a_k \in \A$, we have $v_i(X_i \setminus \{c,c'\}) \geq v_i(X_k)$ for some chores $c,c' \in X_i$.

\begin{restatable}[\EF{2} for two agents and interval graph]{theorem}{TwoAgentsEFtwo}
There is an algorithm that, given any \CISP{} instance with two agents with monotone valuations and an arbitrary interval graph, returns an \EF{2} and maximal schedule. Furthermore, this algorithm runs in polynomial time, assuming access to a value query oracle.
\label{thm:Two_Agents_EF2}
\end{restatable}

Towards proving \Cref{thm:Two_Agents_EF2}, we will once again construct a sequence of schedules, only this time with a slight relaxation of maximality, as described in the next lemma.

\begin{restatable}{lemma}{AdjacentSequence}
For any \CISP{} instance with two agents and an interval graph, there exists a sequence of feasible schedules $( \X_1=(R_1,B_1), \X_2=(R_2,B_2), \dots, \X_k=(R_k,B_k))$ for some $k \in \mathbb{N}$ such that 
\begin{enumerate}
\item The last schedule is obtained by swapping
the two bundles in the first schedule. That is, $R_1=B_k$ and $B_1=R_k$.
\item For any $2\leq i \leq k$, the two schedules $\X_i$ and $\X_{i-1}$ are adjacent.
\item For any $1 \leq i \leq k$, the schedule $\X_i$ is either maximal or can be made maximal by including one unassigned chore in either of its two bundles $R_i$ or $B_i$.
\item The sequence $\X_1,\dots,\X_k$ can be constructed in polynomial time.
\end{enumerate}
\label{lem:Adjacent_sequence}
\end{restatable}
\begin{proof} (of \Cref{lem:Adjacent_sequence})
    We will classify each chore as marked or unmarked via the following iterative process: 
    Consider the chores in 
    the increasing 
    order 
    of their finish times (ties are broken lexicographically), denoted as $c_{1}, c_{2}, \dots, c_{m}$.
    A chore $c$ is classified as \emph{unmarked} if and only if its time interval overlaps with two or more chores which
    have earlier finish time than $c$
    and
    which are classified as marked.
    Consequently, the first two chores in this order, $c_1$ and $c_2$, must be classified as marked, as there cannot be two earlier chores overlapping with them. The classification process then proceeds iteratively for the remaining chores.\\
    Let $\mrk \coloneqq \{c_{j_1},c_{j_2},\dots,c_{j_k}\}$ denote the set of marked chores, and $\overline{\mrk}$ denote the set of unmarked chores. For simplicity of notation, we will henceforth drop $c$ from $c_{j_i}$ when referring to a marked chore. Further, we will assume, without loss of generality, that the chores $j_1,j_2,\dots,j_k$ are in the order of increasing finish time.

\begin{figure*}[h] 
  \centering\begin{tikzpicture}
    \def\rectwidth{7cm}
    \def\rectheight{5.5cm}
    
    \draw (0,0) rectangle (\rectwidth, \rectheight);

    \def\startintervals{{5.2,4.3,4.2,2.5,2.3,3,0.8,1,1.2,0.5}}
    \def\endInterval{{6.5,6,5.5,5,4.8,4.5,4,3.8,2.8,2}}
    \def\height{{0.7,1.15,1.61,2.06,2.52,2.97,3.43,3.88,4.34,4.8}}
    \def\labels{{"$j_5$", "", "$j_4$", "", "", "$j_3$", "", "","$j_2$","$j_1$"}}
    \foreach \y in {1, 2, 3, 4, 5, 6, 7, 8, 9, 10}
    {
        \pgfmathsetmacro\startx{\startintervals[\y-1]}
        \pgfmathsetmacro\endx{\endInterval[\y-1]}
        \pgfmathsetmacro\h{\height[\y-1]}
        \pgfmathsetmacro\labeltext{\labels[\y-1]}

        \draw (\startx, \h) -- (\endx, \h);
        \draw (\startx, \h + 0.1) -- (\startx, \h - 0.1);
        \draw (\endx, \h + 0.1) -- (\endx, \h - 0.1);
        
        \node[right] at (\endx - .4, \h+.3) {\labeltext};
    }
    \draw [decorate, decoration={brace, amplitude=4 pt, mirror}] (0.5,4) -- (0.5,3.2);
    \node[right] at (0.2,3) {\scriptsize $\overline{\mrk}_2$};
    \draw [decorate, decoration={brace, amplitude=4 pt, mirror}] (2.1,2.7) -- (2.1,1.9);
    \node[right] at (1.8,1.7) {\scriptsize $\overline{\mrk}_3$};
    \node[right] at (3.2,1.15) {\scriptsize $\overline{\mrk}_4$};
\end{tikzpicture} 
  \caption{Pictorial representation of marked and unmarked chores classification.}
  \label{fig_marked}
\end{figure*}

    Note that there can be multiple unmarked chores between two consecutive marked chores. For every $i \in \{1,\dots,k\}$, we let $\overline{\mrk}_i$ denote the set of unmarked chores whose finish time is between that of the marked chores $j_i$ and $j_{i+1}$. We let $j_{k+1}$ denote a virtual chore that starts after all other chores have finished; thus, the set $\overline{\mrk}_{k}$ denotes the unmarked chores that finish after $j_k$. \Cref{fig_marked} illustrates an example of marked and unmarked chores.
    
    Now, we will construct the desired sequence of schedules. 
    The unmarked chores will not be assigned to any agent in any of the schedules. 
    We define the $k$ schedules $\X_i = (R_i,B_i)$ for $1\leq i \leq k$, as follows.

\begin{itemize}
\item   $R_1 =\{{j_h} : h \text{ is odd} \}$ and $B_1 = \{{j_h} : h \text{ is even}\}$.

\item For any  $i \in \{2,3, \dots, k-1\}$, 

 \begin{itemize}
\item $R_i$ contains
$\{{j_h} : 1 \leq h < i, h \text{ is even} \} \cup 
\{{j_h} : i < h \leq k, h \text{ is odd} \}$,

\item $B_i$ contains
$\{{j_h} : 1 \leq h < i, h \text{ is odd} \} \cup 
\{{j_h} : i < h \leq k, h \text{ is even} \}$,

\item and the chore ${j_i}$ is assigned or unassigned in the schedule $\X_i$  depending on whether the interval of ${j_i}$ intersects with the intervals of ${j_{i-1}}$ and ${j_{i+1}}$.
Note that one of chores ${j_{i-1}}$ and ${j_{i+1}}$ is in $R_i$ and the other one is in $B_i$. 
If the interval of ${j_i}$ intersects with both, then ${j_i}$ is unassigned. If it intersects with only ${j_{i+1}}$, then
it is assigned to the same agent as ${j_{i-1}}$. 
And otherwise, it is assigned to the same agent as
${j_{i+1}}$.
 \end{itemize}
 
 \item $R_k =\{{j_h} : h \text{ is even} \}$ and $B_k = \{{j_h} : h \text{ is odd}\}$.

\end{itemize}

\emph{Feasibility:} By the definition of marked chores, any marked chore ${j_h}$ can intersect with at most one marked chore with an earlier finish time, which can be only ${j_{h-1}}$. 
But, the two chores ${j_h}$ and ${j_{h-1}}$ 
are not assigned to the same agent in any of the schedules, if their intervals intersect. 
Moreover, the unmarked chores are unassigned in all
the schedules.
Hence, each schedule $\X_i$ is feasible.

\smallskip
\emph{Proof of part (1)}: 
From the construction, $R_1=B_k$ and $B_1=R_k$ immediately follows. 

\smallskip
\emph{Proof of part (2)}:
For any $i\geq 2$, observe that the only chores that can possibly be in $R_i \setminus R_{i-1}$ are ${j_{i-1}}$ and ${j_{i}}$. Moreover, if  ${j_{i-1}}$ is in $R_i \setminus R_{i-1}$, then $i-1$ must be even. In that case, ${j_i}$ must be in $R_{i-1}$, and hence cannot be in $R_i \setminus R_{i-1}$. Thus, we can conclude that $\abs{R_i \setminus R_{i-1}} \leq 1$.

Similarly, the only chores that can possibly be in $R_{i-1} \setminus R_{i}$ are ${j_{i-1}}$ and ${j_{i}}$. Moreover, if ${j_{i}}$ is in $R_{i-1}$, then $i$ must be odd. In that case, ${j_{i-1}}$ must be in $R_i$, and hence cannot be in $R_{i-1} \setminus R_{i}$. Thus, we can conclude that $\abs{R_{i-1} \setminus R_{i}} \leq 1$.

Similarly, one can argue that $\abs{B_i \setminus B_{i-1}} \leq 1$ and $\abs{B_{i-1} \setminus B_{i}} \leq 1$. Thus, the two schedules $\X_{i-1}$ and $\X_i$ are adjacent for any $i\geq 2$.

\smallskip
\emph{Proof of part (3)}:
Consider the schedule $\X_i$ for some $i$. We will argue that there is at most one unassigned chore that can be assigned to some agent without violating feasibility. First, consider the chore ${j_i}$. It is unassigned only if it intersects with two chores, one in $R_i$ and the other in $B_i$. Hence, it cannot be assigned to any agent.

The other unassigned chores are the unmarked chores.
Note that $\overline{\mrk}_1$ is empty.
By the design of the marked and unmarked scheme, any chore in set $\overline{\mrk}_h$  intersects with ${j_{h-1}}$ and ${j_{h}}$, and possibly some other marked chores with earlier finish time.
For any $h$ different from $i$ and $i+1$, observe that chores ${j_{h-1}}$ and ${j_{h}}$ have been assigned to different agents in the schedule $\X_i$.
Hence, any chore in $\overline{\mrk}_h$ cannot be assigned to any agent for $h\neq i$ and $h \neq i+1$. 
      
It remains to consider the unmarked chores in the sets $\overline{\mrk}_i$ and $\overline{\mrk}_{i+1}$. %
We first claim that any chore in $\overline{\mrk}_i$ cannot be assigned to any agent in the schedule $\X_i$. There are two cases: 

(i) when ${j_i}$ and ${j_{i+1}}$ do not intersect, then they both are assigned to the same agent, while ${j_{i-1}}$ is assigned to the other agent. Thus, we get that any chore in $\overline{\mrk}_i$ intersects with two chores which are with different agents, and hence cannot be assigned to any agent.

(ii) when ${j_i}$ and ${j_{i+1}}$ intersect, then ${j_{i+1}}$ must also intersect with every chore in $\overline{\mrk}_i$, because chores in $\overline{\mrk}_i$ have later finish time than ${j_i}$. In this case, any chore in $\overline{\mrk}_i$ intersects with ${j_{i-1}}$ and ${j_{i+1}}$ which are with different agents. Hence, it cannot be assigned to any agent. 

Now, we consider the chores in $\overline{\mrk}_{i+1}$. Any chore in $\overline{\mrk}_{i+1}$ intersecting with ${j_{i-1}}$ cannot be assigned to any agent, because ${j_{i-1}}$ and ${j_{i+1}}$ are with different agents. In case ${j_i}$ and ${j_{i+1}}$ are assigned to different agents, we can say that no chore in $\overline{\mrk}_{i+1}$ can be assigned to any agent. 
      
Consider the case when ${j_{i}}$ and ${j_{i+1}}$ are with the same agent or if ${j_i}$ is unassigned. Let $j$ be the chore in $\overline{\mrk}_{i+1}$ with the earliest finish time that does not intersect with ${j_{i-1}}$.       
If $j$ intersects with ${j_{i+2}}$, then so do the all the following chores in $\overline{\mrk}_{i+1}$. In that case, the schedule is already maximal. Otherwise, we can assign $j$ to the same agent as ${j_{i+2}}$. 
After this, the remaining chores in $\overline{\mrk}_{i+1}$ cannot be assigned to any agent because they all intersect with ${j_{i+1}}$ and $j$. 
Hence, we will end up with a maximal schedule after assigning $j$.
\end{proof}

From \Cref{lem:Adjacent_sequence}, we obtain a feasible schedule which is \EF{1} and can be made maximal by assigning one of its unassigned chores to some agent. The resulting schedule, therefore, satisfies \EF{2}, completing the proof of \Cref{thm:Two_Agents_EF2}.

\subsection{EF1 for two agents and interval graph}

By a more sophisticated coloring argument, it can be shown that an \EF{1} and maximal schedule always exists for interval graphs. We present this proof in \Cref{sec:Proof_Two_Agents_EF1_Interval}. As mentioned in the introduction, \cite{IMY25dividing} show a more general version of this result with a simpler and more intuitive proof. We recommend that the reader refer to their proof. Specifically, their proof further refines our notion of adjacency (\cref{def:adjacent}) and isolates precisely the properties we need, resulting in a cleaner overall argument.

\begin{restatable}[\EF{1} for two agents and interval graph]{theorem}{TwoAgentsEFoneInterval}
There is an algorithm that, given any \CISP{} instance with two agents with monotone valuations and an arbitrary interval graph, returns an \EF{1} and maximal schedule. Furthermore, this algorithm runs in polynomial time assuming access to a value query oracle.
\label{thm:Two_Agents_EF1_Interval}
\end{restatable}

\section{Results for Arbitrary Number of Agents}
\label{sec:Results_n_Agents}

We now turn to the case of an arbitrary number of agents. Here, our main result is that for all $n \geq 3$ agents with identical additive valuations and an arbitrary path graph, an \EF{1} and complete (therefore, maximal) schedule always exists.

Later, we show that if $n \geq 4$ and we additionally assume that the values our dichotomous, such a schedule can be found in polynomial time. This is done in \cref{sec:dichotomous}. Finally, in \cref{sec:bounded-components}, we consider identical valuations and a graph with bounded connected components~(\Cref{thm:n_Agents_Identical_Bounded_Components}) and show that an \EF{1} and complete (therefore, maximal) schedule always exists and can be computed in polynomial time.

We will start with our result for identical additive valuations and an arbitrary path graph.

\begin{restatable}[Path graph and identical valuations]{theorem}{nAgentsIdentical}
Given any \CISP{} instance with $n \geq 3$ agents with identical additive valuations and an arbitrary path graph, an \EF{1} and complete (therefore, maximal) schedule always exists. 
\label{thm:n_Agents_Identical}
\end{restatable}

Below, we first discuss the key ideas in the proof of \Cref{thm:n_Agents_Identical} for the case of $n=3$ agents. We then extend this discussion to provide a detailed proof for general $n$.
Our proof will use the ``cycle plus triangles'' theorem (\Cref{thm:Cycle_Plus_Triangles}).

\begin{restatable}[Cycle plus triangles;~\citealp{FS92solution}]{theorem}{CyclePlusTriangles}
For any positive integer $p$, let $H$ be any simple graph on $3p$ vertices whose set of edges is the disjoint union of a Hamiltonian cycle and $p$ pairwise vertex-disjoint triangles. Then, $H$ is $3$-colorable.
\label{thm:Cycle_Plus_Triangles}
\end{restatable}

Interestingly, \citet{E90some} conjectured that any ``cycle plus triangles'' graph $H$ always admits a $3$-coloring. This conjecture was proven by \citet{FS92solution}. 

To see how the ``cycle plus triangles'' theorem applies to the chore interval scheduling problem, consider reindexing the chores on the path graph $G$ so that $c_1$ is the most disliked chore (recall that the valuations are identical), $c_2$ is the next most disliked chore, and so on (ties are broken arbitrarily). Furthermore, for simplicity, let the number of chores be a multiple of $3$ (i.e., $m = 3p$); if not, we can add some zero-valued dummy chores that will be removed later.

We now construct the graph $H$. For every chore $c_1,\dots,c_m$, we add a vertex in $H$. We group the chores into triples $\{c_1,c_2,c_3\}, \{c_4,c_5,c_6\}$, and so on, and add edges between every pair of chores in each triple resulting in vertex disjoint triangles. Thus, for $i \in [p]$, the $i^\text{th}$ triangle contains the $(3i-2)^\textup{th}$, $(3i-1)^\textup{th}$, and $(3i)^\textup{th}$ most disliked chores. Further, between any pair of chores $c_i$ and $c_j$ that are adjacent along the path graph $G$, we add an edge in $H$. However, if an edge between $c_i$ and $c_j$ already exists, 
then we connect them via a \emph{nine-vertex gadget}; see \Cref{fig:CyclePlusTriangle} in the supplementary material. The gadget helps in avoiding multiedges between vertices by connecting them via a loop of dummy vertices. Finally, we add an edge connecting the extreme chores on the path graph to complete the Hamiltonian cycle---again, we do this via a gadget if the chores already have an edge between them. Observe that $H$ is a supergraph of $G$.

The graph $H$ constructed above satisfies the conditions of the ``cycle plus triangles'' theorem (\Cref{thm:Cycle_Plus_Triangles}) and is therefore $3$-colorable. 
The set of vertices assigned to each color constitutes an independent set of $H$. By treating each color as an agent, we obtain a feasible and complete schedule.

To argue \EF{1}, observe that for each $i \in [p]$, each agent receives exactly one chore from its $(3i-2)^\textup{th}$, $(3i-1)^\textup{th}$, and $(3i)^\textup{th}$ most disliked chores. By removing the worst chore in its bundle, each agent prefers its own bundle over that of any other agent in an \emph{item for item} manner. By additivity, it follows that the schedule is \EF{1}.

The above construction for $n=3$ agents can be extended to a general number of agents. To this end, we will use a generalization of the ``cycle plus triangles'' theorem called the ``cycle plus $n$-cliques'' theorem; the latter applies to graphs consisting of a 
Hamiltonian cycle and vertex disjoint cliques each with at most $n$ vertices. Such graphs are known to always admit an $n$-coloring~\citep{FS97some}. 

\begin{restatable}[Cycle plus $n$-cliques;~\citealp{FS97some}]{theorem}{CyclePlusNCliques}
Let $H$ be any simple graph whose set of edges is the disjoint union of a Hamiltonian cycle and $p$ pairwise vertex-disjoint complete subgraphs each on at most $n$ vertices. Then, $H$ is $n$-colorable.
\label{thm:Cycle_Plus_N_Cliques}
\end{restatable}

\Cref{thm:Cycle_Plus_N_Cliques} is due to \citet{FS97some}. It extends an earlier result called the ``cycle plus triangles'' theorem by the same authors~\citep{FS92solution}, which was stated as \Cref{thm:Cycle_Plus_Triangles}.

In the remainder of this section, we will discuss a procedure to construct a ``cycle plus $n$-cliques'' graph $H$ from the given instance of the chore interval scheduling problem~(\CISP{}) with $n \geq 3$ agents. We will show that any $n$-coloring in the graph $H$ gives an \EF{1} and maximal schedule in the \CISP{} instance. This reduction will establish the desired implication in \Cref{thm:n_Agents_Identical}.

\paragraph{Reducing \CISP{} to ``cycle plus $n$-cliques''.}  We are given that the conflict graph in the given \CISP{} instance is a path graph. Let us refer to this conflict graph as $G$. Since the valuations are identical, we can reindex the chores so that $c_1$ is the most disliked chore, $c_2$ is the next most disliked chore (breaking ties arbitrarily), and so on. For convenience, we can make the number of chores $m$ an integer multiple of $n$ by adding up to $n-1$ dummy chores, each with a value of $0$. These dummy chores will be helpful in applying \Cref{thm:Cycle_Plus_N_Cliques} to the reduced graph $H$ and will be removed later. For an illustration of our reduction with $n=3$ agents, see \Cref{fig:CyclePlusTriangle}.

\begin{figure*}[t]
\centering
    \begin{subfigure}[b]{0.9\linewidth}
    \centering
    \begin{tikzpicture}
            \tikzset{mynode/.style = {shape=circle,draw,inner sep=0pt,minimum size=15pt}}
            \node[mynode] (1) at (0,0) {$c_7$};
            \node[mynode] (2) at (1,0) {$c_{10}$};
            \node[mynode] (3) at (2,0) {$c_9$};
            \node[mynode] (4) at (3,0) {$c_6$};
            \node[mynode] (5) at (4,0) {$c_5$};
            \node[mynode] (6) at (5,0) {$c_2$};
            \node[mynode] (7) at (6,0) {$c_1$};
            \node[mynode] (8) at (7,0) {$c_8$};
            \node[mynode] (9) at (8,0) {$c_4$};
            \node[mynode] (10) at (9,0) {$c_3$};
            \node[mynode,fill=blue!20] (11) at (10,0) {$c_{11}$};
            \node[mynode,fill=blue!20] (12) at (11,0) {$c_{12}$};
            \draw[line width=2pt] (1) -- (2);
            \draw[line width=2pt] (2) -- (3);
            \draw[line width=2pt] (3) -- (4);
            \draw[line width=2pt] (4) -- (5);
            \draw[line width=2pt] (5) -- (6);
            \draw[line width=2pt] (6) -- (7);
            \draw[line width=2pt] (7) -- (8);
            \draw[line width=2pt] (8) -- (9);
            \draw[line width=2pt] (9) -- (10);
            \draw[line width=2pt,dotted] (10) -- (11);
            \draw[line width=2pt,dotted] (11) -- (12);
        \end{tikzpicture}
        \caption{}
        \label{fig:CyclePlusTriangle-path}
        \vspace{0.2in}
    \end{subfigure}
    %
    %
    \begin{subfigure}[b]{0.9\linewidth}
    \centering
    \begin{tikzpicture}
            \tikzset{mynode/.style = {shape=circle,draw,inner sep=0pt,minimum size=15pt}}
            \tikzset{gadgetnode/.style = {shape=rectangle,draw,inner sep=0pt,minimum size=15pt}} 
            \node[mynode] (1) at (1,1.4) {$c_1$};
            \node[mynode] (2) at (0,0) {$c_2$};
            \node[mynode] (3) at (2,0) {$c_3$};
            \node[mynode] (4) at (4,1.4) {$c_4$};
            \node[mynode] (5) at (3,0) {$c_5$};
            \node[mynode] (6) at (5,0) {$c_6$};
            \node[mynode] (7) at (7,1.4) {$c_7$};
            \node[mynode] (8) at (6,0) {$c_8$};
            \node[mynode] (9) at (8,0) {$c_9$};
            \node[mynode] (10) at (10,1.4) {$c_{10}$};
            \node[mynode,fill=blue!20] (11) at (9,0) {$c_{11}$};
            \node[mynode,fill=blue!20] (12) at (11,0) {$c_{12}$};
            \draw (1) -- (2);
            \draw (2) -- (3);
            \draw (1) -- (3);
            \draw (4) -- (5);
            \draw (5) -- (6);
            \draw (4) -- (6);
            \draw (7) -- (8);
            \draw (8) -- (9);
            \draw (7) -- (9);
            \draw (10) -- (11);
            \draw (11) -- (12);
            \draw (10) -- (12);
            \draw[line width=2pt] (7) to[bend left] (10);
            \draw[line width=2pt] (10) to[bend right] (9);
            \draw[line width=2pt] (9) to[bend left] (6);
            \draw[line width=2pt] (5) to[bend left] (2);
            \draw[line width=2pt] (1) to[out=20,in=110,looseness=1.3] (8);
            \draw[line width=2pt] (8) to[bend right] (4);
            \draw[line width=2pt] (4) to[bend right] (3);
            \draw[line width=2pt] (3) to[out=80,in=110] (11);
            \draw[line width=2pt,color=red] (12) to[out=80,in=60] (7);
            %
            %
            %
            \def\x{1.5}
            \def\y{4}
            \node[gadgetnode] (13) at (\x+1,\y+1.4) {$g^1_1$};
            \node[gadgetnode] (14) at (\x+0,\y+0) {$g^1_2$};
            \node[gadgetnode] (15) at (\x+2,\y+0) {$g^1_3$};
            \node[gadgetnode] (16) at (\x+4,\y+1.4) {$g^1_4$};
            \node[gadgetnode] (17) at (\x+3,\y+0) {$g^1_5$};
            \node[gadgetnode] (18) at (\x+5,\y+0) {$g^1_6$};
            \node[gadgetnode] (19) at (\x+7,\y+1.4) {$g^1_7$};
            \node[gadgetnode] (20) at (\x+6,\y+0) {$g^1_8$};
            \node[gadgetnode] (21) at (\x+8,\y+0) {$g^1_9$};
            \draw (13) -- (14);
            \draw (14) -- (15);
            \draw (13) -- (15);
            \draw (16) -- (17);
            \draw (17) -- (18);
            \draw (16) -- (18);
            \draw (19) -- (20);
            \draw (20) -- (21);
            \draw (19) -- (21);
            \draw (13) to[bend left] (19);
            \draw (19) to[bend left] (16);
            \draw (16) to[out=40,in=80,looseness=1.2] (21);
            \draw (21) to[bend left] (18);
            \draw (18) to[bend left] (20);
            \draw (20) to[bend left] (15);
            \draw (15) to[bend left] (17);
            \draw (17) to[bend left] (14);
            %
            %
            %
            \def\x{0.5}
            \def\y{-4}
            \node[gadgetnode] (22) at (\x+1,\y+1.4) {$g^2_1$};
            \node[gadgetnode] (23) at (\x+0,\y+0) {$g^2_2$};
            \node[gadgetnode] (24) at (\x+2,\y+0) {$g^2_3$};
            \node[gadgetnode] (25) at (\x+4,\y+1.4) {$g^2_4$};
            \node[gadgetnode] (26) at (\x+3,\y+0) {$g^2_5$};
            \node[gadgetnode] (27) at (\x+5,\y+0) {$g^2_6$};
            \node[gadgetnode] (28) at (\x+7,\y+1.4) {$g^2_7$};
            \node[gadgetnode] (29) at (\x+6,\y+0) {$g^2_8$};
            \node[gadgetnode] (30) at (\x+8,\y+0) {$g^2_9$};
            \draw (22) -- (23);
            \draw (23) -- (24);
            \draw (22) -- (24);
            \draw (25) -- (26);
            \draw (26) -- (27);
            \draw (25) -- (27);
            \draw (28) -- (29);
            \draw (29) -- (30);
            \draw (28) -- (30);
            \draw (22) to[bend left] (28);
            \draw (28) to[bend left] (25);
            \draw (25) to[out=40,in=80,looseness=1.2] (30);
            \draw (30) to[bend left] (27);
            \draw (27) to[bend left] (29);
            \draw (29) to[bend left] (24);
            \draw (24) to[bend left] (26);
            \draw (26) to[bend left] (23);
            %
            %
            \def\x{10.5}
            \def\y{-4}
            \node[] (31) at (\x+1,\y+1.4) {};
            \node[] (32) at (\x+0,\y+0) {};
            \node[] (33) at (\x+0.5,\y+0.7) {gadget};
            %
            %
            %
            \draw[dashed] (2) to[out=95,in=200,looseness=1.2] (13);
            \draw[dashed] (14) to[bend right,looseness=1.2] (1);
            \draw[dashed] (5) to[bend right] (23);
            \draw[dashed] (22) to[bend left] (6);
            \draw[dashed] (11) to[bend right] (32);
            \draw[dashed] (31) to[bend left] (12);
        \end{tikzpicture}
        \caption{}
        \label{fig:CyclePlusTriangle-construction}
    \end{subfigure}
    %
    %
    %
    %
    \caption{An illustration of the reduction in the proof of \Cref{thm:n_Agents_Identical} for $n=3$ agents. Subfigure (a) shows the path graph consisting of ten original chores $c_7,c_{10},\dots,c_3$ and two dummy chores $c_{11}$ and $c_{12}$ (denoted by shaded nodes). The dummy chores can be considered to be an extension of the path graph as shown. Subfigure (b) shows the corresponding ``cycle plus $n$-cliques'' graph $H$. We divide the chores into groups of size $n$, namely $\{c_1,c_2,c_3\}$, $\{c_4,c_5,c_6\}$, and so on such that $\{c_1,c_2,c_3\}$ are the worst $n$ chores, $\{c_4,c_5,c_6\}$ are the next worst $n$ chores, and so on. The thick black edges simulate the edges of the path graph between chores belonging to different groups. For chores in the same group that are adjacent along the path (i.e., the \emph{special} pairs, namely, $\{c_1,c_2\}$, $\{c_5,c_6\}$, and $\{c_{11},c_{12}\}$), we simulate the edge via a nine-vertex gadget. The square nodes denote the gadget chores. The thin dashed lines denote the connector edges between the special vertices and the gadget vertices. The thick red edge denotes the fictitious edge added between the extreme chores on the path to complete the Hamiltonian cycle. \Cref{fig:CyclePlusTriangle-Coloring} shows a $3$-coloring of the graph $H$.}
\label{fig:CyclePlusTriangle}
\end{figure*}

\begin{figure*}[t]
\centering
    \begin{tikzpicture}
            \tikzset{mynode/.style = {shape=circle,draw,inner sep=0pt,minimum size=15pt}}
            \tikzset{gadgetnode/.style = {shape=rectangle,draw,inner sep=0pt,minimum size=15pt}} 
            \node[mynode] (1) at (1,1.4) {$1$};
            \node[mynode] (2) at (0,0) {$2$};
            \node[mynode] (3) at (2,0) {$3$};
            \node[mynode] (4) at (4,1.4) {$1$};
            \node[mynode] (5) at (3,0) {$3$};
            \node[mynode] (6) at (5,0) {$2$};
            \node[mynode] (7) at (7,1.4) {$1$};
            \node[mynode] (8) at (6,0) {$2$};
            \node[mynode] (9) at (8,0) {$3$};
            \node[mynode] (10) at (10,1.4) {$2$};
            \node[mynode,fill=blue!20] (11) at (9,0) {$1$};
            \node[mynode,fill=blue!20] (12) at (11,0) {$3$};
            \draw (1) -- (2);
            \draw (2) -- (3);
            \draw (1) -- (3);
            \draw (4) -- (5);
            \draw (5) -- (6);
            \draw (4) -- (6);
            \draw (7) -- (8);
            \draw (8) -- (9);
            \draw (7) -- (9);
            \draw (10) -- (11);
            \draw (11) -- (12);
            \draw (10) -- (12);
            \draw[line width=2pt] (7) to[bend left] (10);
            \draw[line width=2pt] (10) to[bend right] (9);
            \draw[line width=2pt] (9) to[bend left] (6);
            \draw[line width=2pt] (5) to[bend left] (2);
            \draw[line width=2pt] (1) to[out=20,in=110,looseness=1.3] (8);
            \draw[line width=2pt] (8) to[bend right] (4);
            \draw[line width=2pt] (4) to[bend right] (3);
            \draw[line width=2pt] (3) to[out=80,in=110] (11);
            \draw[line width=2pt,color=red] (12) to[out=80,in=60] (7);
            %
            %
            %
            \def\x{1.5}
            \def\y{4}
            \node[gadgetnode] (13) at (\x+1,\y+1.4) {$1$};
            \node[gadgetnode] (14) at (\x+0,\y+0) {$2$};
            \node[gadgetnode] (15) at (\x+2,\y+0) {$3$};
            \node[gadgetnode] (16) at (\x+4,\y+1.4) {$2$};
            \node[gadgetnode] (17) at (\x+3,\y+0) {$1$};
            \node[gadgetnode] (18) at (\x+5,\y+0) {$3$};
            \node[gadgetnode] (19) at (\x+7,\y+1.4) {$3$};
            \node[gadgetnode] (20) at (\x+6,\y+0) {$2$};
            \node[gadgetnode] (21) at (\x+8,\y+0) {$1$};
            \draw (13) -- (14);
            \draw (14) -- (15);
            \draw (13) -- (15);
            \draw (16) -- (17);
            \draw (17) -- (18);
            \draw (16) -- (18);
            \draw (19) -- (20);
            \draw (20) -- (21);
            \draw (19) -- (21);
            \draw (13) to[bend left] (19);
            \draw (19) to[bend left] (16);
            \draw (16) to[out=40,in=80,looseness=1.2] (21);
            \draw (21) to[bend left] (18);
            \draw (18) to[bend left] (20);
            \draw (20) to[bend left] (15);
            \draw (15) to[bend left] (17);
            \draw (17) to[bend left] (14);
            %
            %
            %
            \def\x{0.5}
            \def\y{-4}
            \node[gadgetnode] (22) at (\x+1,\y+1.4) {$3$};
            \node[gadgetnode] (23) at (\x+0,\y+0) {$2$};
            \node[gadgetnode] (24) at (\x+2,\y+0) {$1$};
            \node[gadgetnode] (25) at (\x+4,\y+1.4) {$2$};
            \node[gadgetnode] (26) at (\x+3,\y+0) {$3$};
            \node[gadgetnode] (27) at (\x+5,\y+0) {$1$};
            \node[gadgetnode] (28) at (\x+7,\y+1.4) {$1$};
            \node[gadgetnode] (29) at (\x+6,\y+0) {$2$};
            \node[gadgetnode] (30) at (\x+8,\y+0) {$3$};
            \draw (22) -- (23);
            \draw (23) -- (24);
            \draw (22) -- (24);
            \draw (25) -- (26);
            \draw (26) -- (27);
            \draw (25) -- (27);
            \draw (28) -- (29);
            \draw (29) -- (30);
            \draw (28) -- (30);
            \draw (22) to[bend left] (28);
            \draw (28) to[bend left] (25);
            \draw (25) to[out=40,in=80,looseness=1.2] (30);
            \draw (30) to[bend left] (27);
            \draw (27) to[bend left] (29);
            \draw (29) to[bend left] (24);
            \draw (24) to[bend left] (26);
            \draw (26) to[bend left] (23);
            %
            %
            \def\x{10.5}
            \def\y{-4}
            \node[] (31) at (\x+1,\y+1.4) {};
            \node[] (32) at (\x+0,\y+0) {};
            \node[] (33) at (\x+0.5,\y+0.7) {gadget};
            %
            %
            %
            \draw[dashed] (2) to[out=95,in=200,looseness=1.2] (13);
            \draw[dashed] (14) to[bend right,looseness=1.2] (1);
            \draw[dashed] (5) to[bend right] (23);
            \draw[dashed] (22) to[bend left] (6);
            \draw[dashed] (11) to[bend right] (32);
            \draw[dashed] (31) to[bend left] (12);
        \end{tikzpicture}
        \caption{A $3$-coloring of the graph $H$ shown in \Cref{fig:CyclePlusTriangle}.}
\label{fig:CyclePlusTriangle-Coloring}
\end{figure*}

Let $m$ denote the total number of chores in the path graph after adding the dummy chores. The ``cycle plus $n$-cliques'' graph $H$ is constructed as follows: For each chore, we will create a vertex in $H$. To define the edges in $H$, the chores are grouped into $m/n$ groups, formed by putting together blocks of $n$ worst chores as per agents' preferences. That is, the groups are $\{c_1,c_2,\dots,c_n\}$, $\{c_{n+1},c_{n+2},\dots,c_{2n}\}$, and so on. Each group corresponds to an $n$-clique in $H$, meaning we will add an edge between every pair of vertices corresponding to the chores in the same group. In \Cref{fig:CyclePlusTriangle}, this construction is illustrated in the form of four triangles $\{c_1,c_2,c_3\},\{c_4,c_5,c_6\},\{c_7,c_8,c_9\}$ and $\{c_{10},c_{11},c_{12}\}$ in subfigure (b).

Next, we will incorporate the edges in the path graph $G$ into the ``cycle plus $n$-cliques'' graph $H$. For any pair of chores that are adjacent in $G$, we will add an edge between the corresponding vertices in $H$ (shown via thick black edges in \Cref{fig:CyclePlusTriangle}). 
However, there may be chores that are adjacent in the path graph $G$ that belong to the same group; we call such pairs \emph{special}. In \Cref{fig:CyclePlusTriangle}, the pairs $\{c_1,c_2\}$, $\{c_5,c_6\}$, and $\{c_{11},c_{12}\}$ correspond to the special pairs. To connect the vertices corresponding to any special pair of chores, we use a \emph{gadget} (\Cref{fig:CyclePlusTriangle} shows two such gadgets). For any $i$, the $i^\textup{th}$ gadget consists of nine vertices $g^i_1,g^i_2,\dots,g^i_9$. These vertices are connected as follows: First, we group the gadget vertices into three triples $\{g^i_1,g^i_2,g^i_3\},\{g^i_4,g^i_5,g^i_6\}$, and $\{g^i_7,g^i_8,g^i_9\}$ and connect all vertex pairs within the same triple. Next, we add eight more edges, namely $\{g^i_1,g^i_7\}$, $\{g^i_7,g^i_4\}$, $\{g^i_4,g^i_9\}$, $\{g^i_9,g^i_6\}$, $\{g^i_6,g^i_8\}$, $\{g^i_8,g^i_3\}$, $\{g^i_3,g^i_5\}$, and $\{g^i_5,g^i_2\}$. Note that with this construction, the subgraph induced by the gadget consists of a union of pairwise vertex-disjoint triangles and a Hamiltonian cycle except for the missing edge between $g^i_1$ and $g^i_2$. We connect the gadget vertices $g^i_1$ and $g^i_2$ with the vertices corresponding to the $i^\textup{th}$ special pair of chores using two \emph{connector} edges (shown in \Cref{fig:CyclePlusTriangle}, subfigure (b) via thin dashed lines). 
It is relevant to note that for general $n$, while the original chores will induce $n$-cliques, the gadget will nevertheless consist only of triangles.

Finally, we add one more edge between the vertices corresponding to the start and the end chores of the path (including the dummy chores). In \Cref{fig:CyclePlusTriangle}, this edge is shown via a thick red edge between $c_7$ and $c_{12}$. (Again, if an edge already exists between these chores, we connect them through another gadget.) This finishes the construction of the graph~$H$.

\begin{proof} (of \Cref{thm:n_Agents_Identical})
    It is easy to see that the edge set of the graph $H$ consists of a disjoint union of pairwise vertex disjoint $n$-cliques and a Hamiltonian cycle. By the ``cycle plus $n$-cliques'' theorem~(\Cref{thm:Cycle_Plus_N_Cliques}), the graph $H$ admits an $n$-coloring. Note that since $H$ is a supergraph of $G$, an $n$ coloring of $H$ induces an $n$-coloring of $G$.

    Consider any bijection between colors and agents. Assign each chore in $G$ (including dummy chores) to the agent whose associated color is assigned to the corresponding vertex. Since the coloring is proper, no pair of adjacent vertices is assigned to the same color; thus, the assignment is feasible. Furthermore, since all chores are allocated, the schedule is complete and, therefore, maximal.

    To see why the schedule is \EF{1}, observe that each agent receives exactly one chore from the vertices corresponding to each $n$-clique. Fix a pair of agents $h$ and $k$. For any $t \in \{1,2,\dots,m/n\}$, denote the chores in group $t$ as $\{c_{(t-1)n + 1},\dots,c_{(t-1)n + n}\}$. Observe that for any $t \geq 2$, agent $h$ prefers the chore it receives from the $t^\textup{th}$ group over that received by agent $k$ from the $(t-1)^\textup{th}$ group. By additivity, it follows that agent $h$ prefers its bundle (after the removal of the chore from the group $t=1$) over a subset of the bundle of agent $k$. Furthermore, since chores have nonpositive value, agent $h$ prefers its bundle over that of agent $k$ up to the removal of some chore, which implies \EF{1}.
    
    Finally, observe that the vertices corresponding to all dummy chores are located in the same $n$-clique. Since all vertices within each $n$-clique are assigned distinct colors, no agent receives more than one dummy chore. As the dummy chores have zero value, their removal does not affect the \EF{1} guarantee. This proves \Cref{thm:n_Agents_Identical}.
\end{proof}

Finally, we note that for $n=4$ agents, a polynomial-time algorithm is known for finding the $n$-coloring whose existence is guaranteed by the ``cycle plus $n$-cliques'' theorem; see the survey by \citet{A93restricted} and the references therein. The algorithmic problem for other values of $n$ remains open.

\subsection{Polynomial-time Algorithm for Dichotomous Valuations}
\label{sec:dichotomous}

We now show that when $n \geq 4$ and the valuations are additionally assumed to be dichotomous, we can prove an analog of \cref{thm:n_Agents_Identical} where the schedule can be computed in polynomial time. Recall that under identical, additive, and dichotomous valuations, each chore is either ``heavy'' (highly disliked) or ``light'' (less disliked) for all agents. We will assume that heavy chores are valued at $H$ while the light chores are valued at $L$ by each agent; thus, $H < L \leq 0$. 

\begin{restatable}[Path graph and identical dichotomous vals]{theorem}{nAgentsIdenticalDichotomous}
For any $n \geq 4$, there is a polynomial-time algorithm that, given any \CISP{} instance with $n$ agents with identical, dichotomous, and additive valuations and an arbitrary path graph, returns an \EF{1} and maximal schedule.
\label{thm:n_Agents_Identical_Dichotomous}
\end{restatable}

\begin{algorithm}[t]
 \DontPrintSemicolon
\KwIn{A \CISP{} instance $\langle \A, \C, \T, \V \rangle$ with a path graph}
\KwOut{A schedule $\X$}
\BlankLine
\Comment{Phase 1: Grouping agents into meta agents}
\BlankLine
   If $n$ is even: For $i \in [n/2]$, let the meta agent $S_i$ represent the pair of agents $a_{2i-1}$ and $a_{2i}$.\;
   If $n$ is odd: Let the meta agent $S_1$ represent the triple of agents $a_1,a_2,a_3$ and for $i \in \{2,3,\dots,(n-1)/2\}$, let the meta agent $S_i$ represent the agents $a_{2i}$ and $a_{2i+1}$.\;
\BlankLine
\BlankLine
\Comment{Phase 2: Weighted round robin for meta agents}
\BlankLine
If $n$ is even: First allocate heavy chores from left to right via round robin among the meta agents; specifically, we use the picking sequence $\langle S_1,S_2,\dots,S_{n/2},S_1,S_2,\dots,S_{n/2} \rangle$. Then allocate the light chores from left to right according to the same sequence, starting after the meta agent who last picked a heavy chore.\label{algline:Even_Sequence}\;
If $n$ is odd: First allocate heavy chores from left to right via weighted round robin among meta agents according to the picking sequence $\langle S_1,S_2,\dots,S_{(n-1)/2},S_1,S_2,\dots,S_{(n-1)/2},S_1 \rangle$. Then allocate light chores from left to right according to the same sequence, starting after the meta agent who last picked a heavy chore.\label{algline:Odd_Sequence}\;
Assign dummy chores (by adding isolated vertices in the conflict graph) to equalize the number of heavy (similarly, light) chores for all meta agents representing pairs of agents. If $n$ is odd, the meta agent $S_1$ receives $1.5$ times as many chores of each type.
\BlankLine
\BlankLine
\Comment{Phase 3: Assigning chores to individual agents}
\BlankLine
Solve the 2-agent and 3-agent sub-problems by giving each meta agent's chores to its constituent agents.\;
Remove the dummy chores.\;
\KwRet{Current schedule}\;
\caption {Algorithm for Identical Dichotomous Valuations}
\label{alg:Identical_Dichotomous}
\end{algorithm}

Our algorithm for establishing \Cref{thm:n_Agents_Identical_Dichotomous} consists of three phases (see Algorithm~\ref{alg:Identical_Dichotomous}):

\paragraph{Phase 1: Grouping agents into meta agents}
The algorithm starts by partitioning the $n$ agents into pairs (if $n$ is even) or pairs and one triple (if $n$ is odd). We refer to each such group as a \emph{meta agent}.

\paragraph{Phase 2: Weighted round robin for meta agents}
In the second phase, the algorithm runs a \emph{weighted round-robin} procedure for the meta agents. If the number of agents $n$ is even, the procedure is the same as the standard round robin algorithm.\footnote{For ease of analysis, we execute each ``round'' as a combination of two copies of the same permutation of the agents, i.e., using the picking sequence $\langle S_1,S_2,\dots,S_{n/2},S_1,S_2,\dots,S_{n/2} \rangle$, where $S_1,S_2,\dots,S_{n/2}$ are the meta agents. Doing so ensures that each meta agent picks an even number of chores in each round.}
However, if $n$ is odd, we need to adjust the turn-taking so that the meta agent corresponding to the triple picks 1.5 times as many chores as any other meta agent (see Line~\ref{algline:Odd_Sequence} of Algorithm~\ref{alg:Identical_Dichotomous}).\footnote{We use the picking sequence $\langle S_1,S_2,\dots,S_{(n-1)/2},S_1,S_2,\dots,S_{(n-1)/2},S_1 \rangle$ among the meta agents, which ensures that the number of chores picked by each meta agent representing a pair of agents is a multiple of $2$ while the number of chores picked by meta agent $S_1$ is a multiple of $3$.}

The weighted round-robin step runs in two phases: A \emph{heavy phase}, where, at its turn, each meta agent picks the leftmost available heavy chore. This is followed by a \emph{light phase}, where, starting after the meta agent who was the last to pick in the previous phase, each meta agent picks the leftmost available light chore. Note that it is possible for a meta agent to pick two neighboring chores in the path graph.

By adding an appropriate number of \emph{dummy} chores (heavy and light), the algorithm ensures that each meta agent representing a pair of agents gets the same number of heavy (similarly, light) chores, and the meta agent representing a triple gets 1.5 times as many chores of each type.

\Cref{fig:WeightedRoundRobin_Even,fig:WeightedRoundRobin_Odd} provide a visual illustration of the weighted round-robin step in our algorithm (Algorithm~\ref{alg:Identical_Dichotomous}).

\begin{figure*}
    \begin{tikzpicture}
	\node[] (1) at (0,-1) {\Large{$S_1$}};
        \node[] (2) at (2.4,-1) {\Large{$S_i$}};
        \node[] (2a) at (3.6,-1) {\Large{$S_{i+1}$}};
        \node[] (3) at (5.4,-1) {\Large{$S_{n/2}$}};
        \node[] (4) at (6.6,-1) {\Large{$S_1$}};
        \node[] (5) at (8.4,-1) {\Large{$S_i$}};
        \node[] (5a) at (9.6,-1) {\Large{$S_{i+1}$}};
        \node[] (6) at (12,-1) {\Large{$S_{n/2}$}};
        \path (1) -- node[auto=false]{\dots} (2);
        \path (2a) -- node[auto=false]{\dots} (3);
        \path (4) -- node[auto=false]{\dots} (5);
        \path (5a) -- node[auto=false]{\dots} (6);

        \draw[ultra thick][fill=gray!40] (-0.4,-2.75) rectangle (0.4,-1.95);
        \draw[ultra thick][fill=gray!40] (2,-2.75) rectangle (2.8,-1.95);
        \draw[ultra thick][fill=gray!40] (3.2,-2.75) rectangle (4,-1.95);
        \draw[ultra thick][fill=gray!40] (5,-2.75) rectangle (5.8,-1.95);
        \draw[ultra thick][fill=gray!40] (6.2,-2.75) rectangle (7,-1.95);
        \draw[ultra thick][fill=gray!40] (8,-2.75) rectangle (8.8,-1.95);
        \draw[ultra thick][fill=gray!40] (9.2,-2.75) rectangle (10.0,-1.95);
        \draw[ultra thick][fill=gray!40] (11.6,-2.75) rectangle (12.4,-1.95);

        \path (0, -2.25) -- node[auto=false]{\dots} (2.4,-2.25);
        \path (3.6, -2.25) -- node[auto=false]{\dots} (5.4,-2.25);
        \path (6, -2.25) -- node[auto=false]{\dots} (9,-2.25);
        \path (9.6, -2.25) -- node[auto=false]{\dots} (12,-2.25);

        \draw[thick][fill=gray!40] (-0.4,-4) rectangle (0.4,-3.2);
        \draw[thick][fill=gray!40] (2,-4) rectangle (2.8,-3.2);
        \draw[thick, dashed][fill=gray!10] (3.2,-4) rectangle (4,-3.2);
        \draw[thick, dashed][fill=gray!10] (5,-4) rectangle (5.8,-3.2);
        \draw[thick, dashed][fill=gray!10] (6.2,-4) rectangle (7,-3.2);
        \draw[thick, dashed][fill=gray!10] (8,-4) rectangle (8.8,-3.2);
        \draw[thick, dashed][fill=gray!10] (9.2,-4) rectangle (10.0,-3.2);
        \draw[thick, dashed][fill=gray!10] (11.6,-4) rectangle (12.4,-3.2);

        \path (0, -3.5) -- node[auto=false]{\dots} (2.4,-3.5);
        \path (3.6, -3.5) -- node[auto=false]{\dots} (5.4,-3.5);
        \path (6, -3.5) -- node[auto=false]{\dots} (9,-3.5);
        \path (9.6, -3.5) -- node[auto=false]{\dots} (12,-3.5);

        \draw[ultra thick][fill=gray!40] (3.35,-5.25) rectangle (3.85,-4.75);
        \draw[ultra thick][fill=gray!40] (5.15,-5.25) rectangle (5.65,-4.75);
        \draw[ultra thick][fill=gray!40] (6.35,-5.25) rectangle (6.85,-4.75);
        \draw[ultra thick][fill=gray!40] (8.15,-5.25) rectangle (8.65,-4.75);
        \draw[ultra thick][fill=gray!40] (9.35,-5.25) rectangle (9.85,-4.75);
        \draw[ultra thick][fill=gray!40] (11.75,-5.25) rectangle (12.25,-4.75);

        \path (3.6, -5) -- node[auto=false]{\dots} (5.4,-5);
        \path (6.6, -5) -- node[auto=false]{\dots} (8.4,-5);
        \path (9.6, -5) -- node[auto=false]{\dots} (12,-5);

        \draw[ultra thick][fill=gray!40] (-0.25,-6) rectangle (0.25,-5.5);
        \draw[ultra thick][fill=gray!40] (2.15,-6) rectangle (2.65,-5.5);
        \draw[thick][fill=gray!40] (3.35,-6) rectangle (3.85,-5.5);
        \draw[thick][fill=gray!40] (5.15,-6) rectangle (5.65,-5.5);
        \draw[thick][fill=gray!40] (6.35,-6) rectangle (6.85,-5.5);
        \draw[thick][fill=gray!40] (8.15,-6) rectangle (8.65,-5.5);
        \draw[thick, dashed][fill=gray!10] (9.35,-6) rectangle (9.85,-5.5);
        \draw[thick, dashed][fill=gray!10] (11.75,-6) rectangle (12.25,-5.5);
        
        \path (0, -5.75) -- node[auto=false]{\dots} (2.4,-5.75);
        \path (3.6, -5.75) -- node[auto=false]{\dots} (5.4,-5.75);
        \path (6.6, -5.75) -- node[auto=false]{\dots} (8.4,-5.75);
        \path (9.6, -5.75) -- node[auto=false]{\dots} (12,-5.75);

        \draw[thick, dashed][fill=gray!10] (-0.25,-6.75) rectangle (0.25,-6.25);
        \draw[thick, dashed][fill=gray!10] (2.15,-6.75) rectangle (2.65,-6.25);
        
        \path (0, -6.5) -- node[auto=false]{\dots} (2.4,-6.5);

        \draw node at (-2, -2.5) {Heavy Phase};
        \draw node at (-2, -6) {Light Phase};
        
    \end{tikzpicture}
    \vspace{0.01in}
    \caption{Illustrating the weighted round-robin algorithm when the number of agents is even. The large and small squares denote the heavy and light chores, respectively. The strongly (respectively, lightly) shaded squares denote the chores that are originally present (respectively, the dummy chores). Thick borders around the squares denote that in that round, all agents receive an original chore, while the squares with regular borders denote the round where some agents picked original chores while others did not because the original chores were consumed.}
    \label{fig:WeightedRoundRobin_Even}
\end{figure*}

\begin{figure*}
    \begin{tikzpicture}
	\node[] (1) at (0,-1) {\Large{$S_1$}};
        \node[] (2) at (2.4,-1) {\Large{$S_i$}};
        \node[] (3) at (3.6,-1) {\Large{$S_{i + 1}$}};
        \node[] (4) at (6,-1) {\Large{$S_1$}};
        \node[] (5) at (8.4,-1) {\Large{$S_i$}};
        \node[] (6) at (9.6,-1) {\Large{$S_{i + 1}$}};
        \node[] (7) at (12,-1) {\Large{$S_1$}};
        \path (1) -- node[auto=false]{\dots} (2);
        \path (3) -- node[auto=false]{\dots} (4);
        \path (4) -- node[auto=false]{\dots} (5);
        \path (6) -- node[auto=false]{\dots} (7);

        \draw[ultra thick][fill=gray!40] (-0.4,-2.75) rectangle (0.4,-1.95);
        \draw[ultra thick][fill=gray!40] (2,-2.75) rectangle (2.8,-1.95);
        \draw[ultra thick][fill=gray!40] (3.2,-2.75) rectangle (4,-1.95);
        \draw[ultra thick][fill=gray!40] (5.6,-2.75) rectangle (6.4,-1.95);
        \draw[ultra thick][fill=gray!40] (8,-2.75) rectangle (8.8,-1.95);
        \draw[ultra thick][fill=gray!40] (9.2,-2.75) rectangle (10.0,-1.95);
        \draw[ultra thick][fill=gray!40] (11.6,-2.75) rectangle (12.4,-1.95);

        \path (0, -2.25) -- node[auto=false]{\dots} (2.4,-2.25);
        \path (3.6, -2.25) -- node[auto=false]{\dots} (6,-2.25);
        \path (6, -2.25) -- node[auto=false]{\dots} (8.4,-2.25);
        \path (9.6, -2.25) -- node[auto=false]{\dots} (12,-2.25);

        \draw[thick][fill=gray!40] (-0.4,-4) rectangle (0.4,-3.2);
        \draw[thick][fill=gray!40] (2,-4) rectangle (2.8,-3.2);
        \draw[thick, dashed][fill=gray!10] (3.2,-4) rectangle (4,-3.2);
        \draw[thick, dashed][fill=gray!10] (5.6,-4) rectangle (6.4,-3.2);
        \draw[thick, dashed][fill=gray!10] (8,-4) rectangle (8.8,-3.2);
        \draw[thick, dashed][fill=gray!10] (9.2,-4) rectangle (10.0,-3.2);
        \draw[thick, dashed][fill=gray!10] (11.6,-4) rectangle (12.4,-3.2);

        \path (0, -3.5) -- node[auto=false]{\dots} (2.4,-3.5);
        \path (3.6, -3.5) -- node[auto=false]{\dots} (6,-3.5);
        \path (6, -3.5) -- node[auto=false]{\dots} (8.4,-3.5);
        \path (9.6, -3.5) -- node[auto=false]{\dots} (12,-3.5);

        \draw[ultra thick][fill=gray!40] (3.35,-5.25) rectangle (3.85,-4.75);
        \draw[ultra thick][fill=gray!40] (5.75,-5.25) rectangle (6.25,-4.75);
        \draw[ultra thick][fill=gray!40] (8.15,-5.25) rectangle (8.65,-4.75);
        \draw[ultra thick][fill=gray!40] (9.35,-5.25) rectangle (9.85,-4.75);
        \draw[ultra thick][fill=gray!40] (11.75,-5.25) rectangle (12.25,-4.75);

        \path (3.6, -5) -- node[auto=false]{\dots} (6,-5);
        \path (6, -5) -- node[auto=false]{\dots} (8.4,-5);
        \path (9.6, -5) -- node[auto=false]{\dots} (12,-5);

        \draw[ultra thick][fill=gray!40] (-0.25,-6) rectangle (0.25,-5.5);
        \draw[ultra thick][fill=gray!40] (2.15,-6) rectangle (2.65,-5.5);
        \draw[thick][fill=gray!40] (3.35,-6) rectangle (3.85,-5.5);
        \draw[thick][fill=gray!40] (5.75,-6) rectangle (6.25,-5.5);
        \draw[thick][fill=gray!40] (8.15,-6) rectangle (8.65,-5.5);
        \draw[thick, dashed][fill=gray!10] (9.35,-6) rectangle (9.85,-5.5);
        \draw[thick, dashed][fill=gray!10] (11.75,-6) rectangle (12.25,-5.5);
        
        \path (0, -5.75) -- node[auto=false]{\dots} (2.4,-5.75);
        \path (3.6, -5.75) -- node[auto=false]{\dots} (6,-5.75);
        \path (6, -5.75) -- node[auto=false]{\dots} (8.4,-5.75);
        \path (9.6, -5.75) -- node[auto=false]{\dots} (12,-5.75);

        \draw[thick, dashed][fill=gray!10] (-0.25,-6.75) rectangle (0.25,-6.25);
        \draw[thick, dashed][fill=gray!10] (2.15,-6.75) rectangle (2.65,-6.25);
        
        \path (0, -6.5) -- node[auto=false]{\dots} (2.4,-6.5);

        \draw node at (-2, -2.5) {Heavy Phase};
        \draw node at (-2, -6) {Light Phase};
        
    \end{tikzpicture}
    \vspace{0.01in}
    \caption{Illustrating the weighted round-robin algorithm when the number of agents is odd. The drawing convention is the same as described in the caption of \Cref{fig:WeightedRoundRobin_Even}.}
    \label{fig:WeightedRoundRobin_Odd}
\end{figure*}

\paragraph{Phase 3: Assigning chores to individual agents}
In the third phase, the algorithm distributes the chores of each meta agent among its constituent agents; in other words, the algorithm solves the 2-agent and 3-agent subproblems. At the end of this phase, each individual agent has the same number of heavy (similarly, light) chores. Finally, the algorithm discards the dummy chores and the resulting schedule is returned.

It is easy to see that the algorithm runs in polynomial time. To argue correctness, we will show that before the dummy chores are removed in Phase 3, each agent gets an equal number of heavy (similarly, light) chores; in other words, before the removal of dummy chores, the schedule is envy-free. 

Note that since there are at least four agents $(n \geq 4)$, there must be at least two meta agents. If $n$ is even, a meta agent never picks two heavy (respectively, light) chores that are adjacent in the path graph (recall that each meta agent picks the \emph{leftmost} available chore). Thus, the set of chores picked by a meta agent must constitute a disjoint union of edges (between one heavy and one light chore) and isolated vertices. \Cref{lem:Two_Agent_Subproblem} formalizes this observation.

\begin{restatable}{lemma}{TwoAgentSubproblem}
For any meta agent $S_i$ that represents a pair of agents, the set of chores picked by $S_i$ is a disjoint union of edges (between one heavy and one light chore) and isolated vertices. Furthermore, these chores can be divided among the two constituent agents such that both agents get the same number of heavy (similarly, light) chores and their bundles are conflict-free.
\label{lem:Two_Agent_Subproblem}
\end{restatable}
\begin{proof}
Since there are at least four agents ($n \geq 4$), there must be at least two meta agents. Thus, no meta agent can pick two adjacent heavy (respectively, light) chores in the path graph. Thus, the set of chores picked by any meta agent, say $S_i$, that represents a pair of agents is a disjoint collection of paths such that each path is of length at most $2$. In other words, the chores picked up by the meta agent $S_i$ consist of isolated vertices and disjoint edges, where each edge connects a heavy and a light chore.

By the choice of picking sequences in our algorithm (Lines~\ref{algline:Even_Sequence} and~\ref{algline:Odd_Sequence} in Algorithm~\ref{alg:Identical_Dichotomous}), any meta agent that represents a pair of agents picks two chores in each round. Thus, the number of heavy (similarly, light) chores in each meta agent's bundle is even. This includes the dummy chores.

Let $x$ denote the number of edges and $y$ denote the number of isolated vertices in the bundle of meta agent $S_i$. We know that $2x+y$ is even, implying that $y$ is even. Let $a_j$ and $a_{j+1}$ denote the constituent agents of $S_i$.

If $x$ is even, then there are $x$ heavy chores each connected to a unique light chore. For $x/2$ of the edges, the heavy chore can be given to $a_j$ and the light chore to $a_{j+1}$. For the other $x/2$ edges, the heavy chore can be given to $a_{j+1}$ and the light chore to $a_{j}$. Among the $y$ isolated chores, there should again be an even number of heavy (respectively, light) chores, which can be equitably distributed between the two agents.

If $x$ is odd, then for $\lfloor x/2 \rfloor$ of the edges, the heavy chore can be given to $a_j$ and the light chore to $a_{j+1}$. For the other $\lceil x/2 \rceil$ edges, the heavy chore can be given to $a_{j+1}$ and the light chore to $a_{j}$. Thus, so far, agent $a_{j+1}$ has an extra heavy chore, and agent $a_j$ has an extra light chore. Now, among the $y$ isolated chores, there should be an odd number of heavy and an odd number of light chores. By giving an extra heavy chore to $a_j$ and an extra light chore to $a_{j+1}$ among the isolated chores, and allocating the remaining chores evenly, once again an equitable distribution of heavy and light chores can be achieved.
\end{proof}

When the number of agents is odd, the conclusion of \Cref{lem:Two_Agent_Subproblem} still holds true for the meta agents representing pairs of agents. However, for the meta agent $S_1$ which represents a triple of agents, the set of chores constitutes a disjoint union of paths each of length at most $4$. Here, a more careful argument is needed to show an equitable distribution of chores among the three constituent agents of the meta agent $S_1$.

\begin{restatable}{lemma}{ThreeAgentSubproblem}
For the meta agent $S_1$ that represents a triple of agents, the set of chores picked by $S_1$ is a disjoint union of paths of length at most $4$. Furthermore, these chores can be divided among the three constituent agents such that all three agents get the same number of heavy (similarly, light) chores and their bundles are conflict-free.
\label{lem:Three_Agent_Subproblem}
\end{restatable}
\begin{proof}
We will use an argument similar to that in the proof of \Cref{lem:Two_Agent_Subproblem}. By the choice of the picking sequence in our algorithm (Line~\ref{algline:Odd_Sequence} in Algorithm~\ref{alg:Identical_Dichotomous}), the meta agent $S_1$ can never pick three consecutive heavy (respectively, light) chores in each round. Thus, the set of chores picked by $S_1$ is a disjoint union of paths each of length at most $4$. Furthermore, the number of heavy (similarly, light) chores in $S_1$'s bundle is a multiple of $3$. This includes the dummy chores.

Let $a_1,a_2,a_3$ denote the constituent agents of $S_1$. We will show that each of these agents can be assigned an equal number of heavy (similarly, light) chores such that no two chores in an agent's bundle are adjacent. Recall that the total number of heavy (similarly, light) chores is a multiple of $3$. Thus, it will suffice to argue that, after feasibly allocating the chores in $i-1$ components, the chores in the $i^\text{th}$ component can also be feasibly assigned such that for any pair of agents, the number of heavy (similarly, light) chores in their bundles differs by at most $1$.

To show that the number of heavy (similarly, light) chores is almost balanced in the manner stated above, let us consider the different components that arise in the meta agent's bundle:

\begin{itemize}
    \item Path of length $3$: Note that any path of length $3$ contains at most two heavy and at most two light chores. We always assign the three chores in such a path to three different agents. 
    \item Path of length $4$: A path of length $4$ must contain two heavy and two light chores. For any arrangement of these chores, it is always possible to assign to some agent one heavy and one light chore, and assign to the other two agents the remaining heavy and light chore.
    \item Path of length $2$: In this case, we assign the two chores to two different agents.
    \item Isolated vertices: The isolated vertices do not pose any constraints on feasibility.
\end{itemize}

To compute the desired assignment of the chores, we follow a component-wise strategy. At each step, we maintain the following invariant: For any pair of agents, the number of heavy (similarly, light) chores in their bundles differs by at most $1$. By case analysis for each component, it can be checked that the chores in that component can be assigned while maintaining this invariant. The lemma follows from the invariant and the multiples-of-$3$ property.
\end{proof}

\Cref{lem:Two_Agent_Subproblem,lem:Three_Agent_Subproblem} show that the two-agent and three-agent subproblems can be solved so that all chores (including the dummy chores) are allocated, i.e., the schedule is complete and hence maximal. Furthermore, each individual agent gets the same number of heavy (similarly, light) chores. Thus, before removing the dummy chores, the schedule is envy-free.

After the removal of dummy chores, the schedule remains \EF{1}. This is because the number of dummy chores assigned to any pair of agents differ by at most $1$, and, moreover, the number of heavy (respectively, light) dummy chores assigned to any pair of agents also differ by at most $1$. Therefore, removing the dummy chores preserves the \EF{1} property. Also, note that removing the dummy chores maintains the completeness of the schedule and therefore its maximality. This proves \Cref{thm:n_Agents_Identical_Dichotomous}.

\subsection{Polynomial-time Algorithm when Connected Components are Bounded}
\label{sec:bounded-components}
Finally, we consider graphs where all connected components have size at most $n$, the number of agents, and show that when the agents have identical, additive valuations, an \EF{1} and maximal schedule exists and can be efficiently computed.

\begin{restatable}[Bounded connected components and identical valuations]{theorem}{nAgentsIdenticalBoundedComponents}
There is a polynomial-time algorithm that, given any \CISP{} instance with $n$ agents with identical additive valuations and an arbitrary conflict graph with each connected component of size at most $n$, returns an \EF{1} and complete (therefore, maximal) schedule.
\label{thm:n_Agents_Identical_Bounded_Components}
\end{restatable}
\begin{proof}
    For each connected component with fewer than $n$ chores, we will add enough zero-valued dummy chores to make the number of chores exactly equal to $n$. The dummy chores will be removed later without affecting the \EF{1} or maximality guarantees.

    Our algorithm will proceed in a component-wise manner, in that it assigns all chores in component $i$ before proceeding to component $i+1$. For each component, the algorithm will assign each of the $n$ chores to a distinct agent. Thus, the resulting allocation is feasible and complete (therefore, maximal).

    Note that due to identical valuations, the envy graph of the partial allocation at any stage must be acyclic. Thus, the envy graph always has a topological ordering. The chores in component $i$ are assigned based on the topological ordering resulting from the partial allocation of the chores in the components $1,2,\dots,i-1$. Specifically, the source agent (i.e., an unenvied agent) is placed at the front of the ordering, while the sink (i.e., an unenvious agent) is placed at the end.

    Consider any pair of agents $h$ and $k$, and let us analyze the envy of agent $h$ towards agent $k$. Note that the initial (empty) allocation is vacuously \EF{1}. The algorithm will maintain the \EF{1} property after the assignment of chores in each component. Indeed, consider the envy of $h$ towards $k$ before and after the allocation of chores in component $i$. If $h$ is behind $k$ in the ordering before the allocation of chores in component $i$, then $h$ does not envy $k$ to begin with. After assigning the chores in component $i$, $h$ may envy $k$ (since $h$ picks a chore after $k$ does); however, the envy can be eliminated by dropped the most recent chore picked by $h$. Thus, \EF{1} is satisfied in this case. On the other hand, if $h$ is ahead of $k$ in the ordering, then it may envy $k$ to begin with, but the envy is bounded by \EF{1} (due to the invariant). Since $h$ now picks a chore before $k$, the envy continues to be bounded by \EF{1}.

    Finally, we will remove the dummy chores. Note that since the dummy chores have zero value, their removal does not affect the \EF{1} guarantee.
\end{proof}

We note that our algorithm resembles that of \citet[Proposition 2]{HH22fair}, who showed a result similar to \Cref{thm:n_Agents_Identical_Bounded_Components} for indivisible \emph{goods}. Notably, their result does not require the valuations to be identical. Their proof uses the fact that for indivisible goods, resolving an \emph{arbitrary} envy cycle preserves the \EF{1} guarantee~\citep{LMM+04approximately}. For chores, however, resolving arbitrary envy cycles can fail to preserve \EF{1}~\citep{BSV21approximate}. Thus, it is not immediately clear how their analysis can be applied to our setting when the valuations are not identical.

\section{Concluding Remarks}
We formulated the problem of fair and efficient interval scheduling of indivisible chores under conflict constraints. Despite the limited applicability of techniques from the unconstrained setting, we showed, by means of a novel coloring technique, that an \EF{1} and maximal schedule can be computed for two agents. 
Resolving the existence of \EF{1} and maximal schedules for three or more \emph{non-identical} agents for path graphs (and, more generally, for interval graphs) is an exciting open problem. It would also be interesting to consider schedules that ``minimize wastage'' (i.e., leave the fewest chores unassigned) or satisfy other notions of fairness such as proportionality, equitability, or maximin fair share.

\section*{Acknowledgments}
We are thankful to anonymous reviewers of AAMAS 2024 for their helpful feedback, and to Amit Kumar for valuable discussions during the early stages of this work. SN gratefully acknowledges the support of a MATRICS grant (MTR/2021/000367) and a Core Research Grant (CRG/2022/009169) from SERB, Govt. of India, a TCAAI grant (DO/2021-TCAI002-009), and a TCS grant (MOU/CS/10001981-1/22-23). RRS acknowledges support of the Department of Atomic Energy, Government  of India, under project no. RTI4014. RV acknowledges support from DST INSPIRE grant no. DST/INSPIRE/04/2020/000107, ANRF grant no. CRG/2022/002621, and Mr. D.P. Gupta Chair Professorship. .

\bibliographystyle{plainnat} 
\bibliography{ms}

\clearpage
\appendix
\begin{center}
    \Large{Technical Appendix}
\end{center}

\section{Proof of Theorem~\ref{thm:Two_Agents_EF1_Interval}}
\label{sec:Proof_Two_Agents_EF1_Interval}

\TwoAgentsEFoneInterval*
To prove \Cref{thm:Two_Agents_EF1_Interval}, it will suffice to show the following result.

\begin{restatable}{lemma}{AdjacentSequenceInterval}
For any \CISP{} instance with two agents and an interval graph, there exists a sequence of feasible schedules $( \X_0=(R_0,B_0), \X_1=(R_1, B_1), \dots, \X_l=(R_l,B_l))$ for some $l \in \mathbb{N}$ such that 
\begin{enumerate}
\item The last schedule is obtained by swapping
the two bundles in the first schedule. That is, $R_0=B_l$ and $B_0=R_l$.
\item For any $1 \leq i \leq l$, the two schedules $\X_i$ and $\X_{i-1}$ are adjacent.
\item For any $0 \leq i \leq l$, the schedule $\X_i$ is maximal.
\item The sequence $\X_0,\dots,\X_l$ can be constructed in polynomial time.
\end{enumerate}
\label{lem:Adjacent_sequence_Interval}
\end{restatable}

\Cref{thm:Two_Agents_EF1_Interval} follows readily from
\Cref{lem:Adjacent_EF1,lem:Adjacent_sequence_Interval}. Thus, in the rest of this section, we will focus on the proof of \Cref{lem:Adjacent_sequence_Interval}.

\begin{proof} (of \Cref{lem:Adjacent_sequence_Interval})
To create the desired sequence of schedules, we employ a three-phase procedure.
In the first phase, we construct a maximal schedule, just as in the Lemma~\ref{lem:Adjacent_sequence}.
The idea is then to gradually move chores from one bundle to another and reach a schedule with swapped bundles.
However, as seen in the proof of Lemma~\ref{lem:Adjacent_sequence}, the straightforward way of moving chores
violates maximality.
Instead, we first do a preparatory phase 2, where we further classify unassigned chores as either supported or unsupported. Unsupported chores are those that create issues with maximality during such straightforward moves. So in this phase, we perform some reassignment such that no unassigned chore remains unsupported.  
This ensures that we can then move chores between bundles in a natural way and maintain maximality,
which is phase 3.
The high level description of this procedure is given in Algorithm~\ref{algo:three_phase}.\\

\textbf{Phase 1}: In Phase 1, we will construct the first and the last schedules, namely $\X_0$ and $\X_l$, of the sequence. The first schedule $\X_0$ is the same as the first schedule $X_1$ of \Cref{lem:Adjacent_sequence}. So we have set of chores $\mrk \coloneqq \{{j_1},{j_2},\dots,{j_k}\}$ and $\overline{\mrk} \coloneqq \overline{\mrk}_1 \cup \overline{\mrk}_2 \dots \cup \overline{\mrk}_k$.\\

We define the first schedule, $\X_0 = (R_0,B_0)$, as follows:
    $$ R_0 =\{{j_h} : h \text{ is odd} \} \text{, } B_0 = \{{j_h} : h \text{ is even}\}$$ 
    $$\text{and $\overline{\mrk}$ are unassigned}.$$ 

Consequently, the last schedule, $\X_l$ must be swapping of the bundles, resulting in:
 $$R_l =\{{j_h} : h \text{ is even} \} \text{, } B_l = \{{j_h} : h \text{ is odd}\}$$ $$\text{and  $\overline{\mrk}$ are unassigned}.$$

We refer to the first schedule $\X_0$ as the \emph{source} schedule and the last schedule $\X_l$ as the \emph{target} schedule. The bundles to which any chore ${j_h}$ belongs in the schedule $\X_0$ and in the schedule $\X_l$ are called the source bundle and the target bundle of ${j_h}$, respectively. 

For a given schedule $\X_i$, chores that are assigned to either the $R_i$ or $B_i$
  bundles are referred to as assigned or colored chores. Conversely, chores not assigned to either bundle in $\X_i$  are referred to as unassigned or uncolored chores. It is important to note the distinction that the classification of chores as `marked' or `unmarked' is defined with respect to the {\em initial schedule} $\X_0$ and remains invariant throughout the process. In contrast, the classification of chores as `assigned/unassigned' or `colored/uncolored' is defined for {\em any arbitrary schedule} $\X_i$ and may change during the process.
  
  In the initial schedule $\X_0$ all marked chores are assigned, and all unmarked chores are unassigned but this need not be true for any arbitrary $\X_i$. The color of an assigned chore is defined as {\color{red}red} if it is assigned to the $R_i$ bundle and {\color{blue}blue} if it is assigned to the $B_i$ bundle in $\X_i$. \Cref{fig1} illustrates the coloring scheme.   

\begin{figure*}[h] 
  \centering\begin{tikzpicture}
    \def\rectwidth{14cm}
    \def\rectheight{5.5cm}
    
    \draw (0,0) rectangle (\rectwidth, \rectheight);
    
    \draw (7,0) -- (7, \rectheight);
    
    \def\startintervals{{5.2,4.2,2.9,3.3,2.3,1,1.2,0.5}}
    \def\endInterval{{6.5,6,5.5,5,4.5,4,3.1,2}}
    \def\height{{1,1.5,2,2.5,3,3.5,4,4.5}}
    \def\labels{{"R", "N", "N", "B", "R", "N", "B", "R"}}
    \def\colors{{"red", "black", "black", "blue", "red", "black", "blue", "red"}}
    \foreach \y in {1, 2, 3, 4, 5, 6, 7, 8}
    {
        \pgfmathsetmacro\startx{\startintervals[\y-1]}
        \pgfmathsetmacro\endx{\endInterval[\y-1]}
        \pgfmathsetmacro\h{\height[\y-1]}
        \pgfmathsetmacro\labeltext{\labels[\y-1]}
        \pgfmathsetmacro\labelcolor{\colors[\y-1]}

        \draw (\startx, \h) -- (\endx, \h);
        \draw (\startx, \h + 0.1) -- (\startx, \h - 0.1);
        \draw (\endx, \h + 0.1) -- (\endx, \h - 0.1);
        
        \node[right, text=\labelcolor] at (\endx - .4, \h+.3) {\labeltext};
    }
    \node[right] at (3, 0.4) {$X_0$};

    \def\labels2{{"B", "N", "N", "R", "B", "N", "R", "B"}}
    \def\colors2{{"blue", "black", "black", "red", "blue", "black", "red", "blue"}}
    \foreach \y in {1, 2, 3, 4, 5, 6, 7, 8}
    {
        \pgfmathsetmacro\startx{\startintervals[\y-1] + 7}
        \pgfmathsetmacro\endx{\endInterval[\y-1] + 7}
        \pgfmathsetmacro\h{\height[\y-1]}
        \pgfmathsetmacro\labeltext{\labels2[\y-1]}
        \pgfmathsetmacro\labelcolor{\colors2[\y-1]}

        \draw (\startx, \h) -- (\endx, \h);
        \draw (\startx, \h + 0.1) -- (\startx, \h - 0.1);
        \draw (\endx, \h + 0.1) -- (\endx, \h - 0.1);
        
        \node[right, text=\labelcolor] at (\endx - .4, \h+.3) {\labeltext};
    }
    \node[right] at (10, 0.4) {$X_l$};
\end{tikzpicture} 
  \caption{Pictorial representation of first and the last schedule of the sequence: Observe that in the schedule $\X_0$, all marked chores are assigned to alternating bundles represented as $R$ and $B$. Additionally, all unmarked chores are unassigned and indicated as $N$ in the figure. In the target schedule, denoted as $\X_l$, all the marked chores will be assigned to the opposite bundle, while unmarked chores will retain their unassigned status.\\}
  \label{fig1}
\end{figure*}

\emph{Feasibility and maximality of $\X_0$}: According to the definition of a marked chore, any chore ${j_h}$ can overlap with at most two other marked chores, namely, ${j_{h-1}}$ and ${j_{h+1}}$ i.e one before it and one after it.
If $h$ is even, then $h-1$ and $h+1$ must be odd, and vice versa.
Therefore, in the schedule $\X_0$, it follows that  ${j_{h-1}}$ and ${j_{h+1}}$ are not in the same bundle as ${j_h}$, affirming the feasibility of schedule $\X_0$. Additionally, every unmarked chore in the set $\overline{\mrk}_h$ must necessarily overlap with chores ${j_h}$ and ${j_{h-1}}$ for some $h$, 
which belong to different bundles in $\X_0$. 
Consequently, $\X_0$ is also maximal. Since $\X_l$ simply swaps the colors among the assigned chores in $\X_0$, it will also be feasible and maximal.\\

\begin{algorithm}[t]
 \DontPrintSemicolon
\KwIn{A \CISP{} instance $\langle \A, \C, \T, \V \rangle$ with an interval graph}
\KwOut{A sequence of schedules $\X_1,\dots,\X_k$}
 \BlankLine
\Comment{Phase 1: Initial Coloring}
\BlankLine
    Order the chores in increasing order of their finish time (break ties lexicographically) $c_1,c_2,\dots,c_m$\;

    Set $P \gets \emptyset$.\;
    \BlankLine
    \tcp{Classify the chores as marked and unmarked}
    \BlankLine
    \For{$i = 1, \dots, m$}{
        \If{$\nexists c_j \neq c_{j'} \in P: c_i$ overlaps with $c_j$, $c_{j'}$}{
            $P \gets P \cup \{ c_i \}$. \;
        }
    }
    Let $p = |P|$ and $P = \{ c_{j_1}, \dots, c_{j_p} \}$. Define the sets:
    \begin{align*}
        R &= R_0 = \{ c_{j_h} \mid h \in [p] \text{~is odd} \} , \\
        B &= B_0 = \{ c_{j_h} \mid h \in [p] \text{~is even} \} , \\
        \X &= \X_0 = (R, B) .
    \end{align*} \;
\Comment{Phase 2: Removal of unsupported chores}
\BlankLine
    \For{$i = m, \dots, 1$}{
        \If{$c_i \notin R \cup B$ is unsupported in $\X$}{
            Update $\X$ 
            according to one of the five conditions mentioned in the detailed description of Phase 2\;
        }
    }
\BlankLine
\BlankLine
\Comment{Phase 3: Swap the colors of unsettled chores}
\BlankLine
    \For{$i = 1, \dots, m$}{
        Pick the first two unsettled chores\;
        Assign the target color to the first chore\;
        Assign a feasible color to the second chore\;
    }
\KwRet{final sequence of schedules}\;
\caption{Algorithm for \EF{1} and maximality for interval graphs}
\label{algo:three_phase}
\end{algorithm}

\textbf{Phase 2}: In this phase, we will construct the first $k$ schedule (where $k$ is the number of marked chores) of the desired sequence of maximal schedules, starting from $\X_0$.
To achieve this, we first classify all unassigned chores as either supported or unsupported. Unsupported chores are those that pose challenges when we move chores straightforward from one bundle to another. To address these challenges, we design specific moves that transition unsupported chores into supported ones by modifying their status between assigned and unassigned as needed. Each move is both feasible and maintains maximality, thereby contributing to the construction of the desired sequence of maximal schedules.

\begin{definition}[Supported chores]
For any given schedule $\X = (R, B)$, we say that a chore $j \notin R \cup B$ is \emph{supported} in $\X$ if it satisfies at least one of the following three conditions:

\begin{enumerate}

\item there exist distinct chores $j^{(1)}, j^{(2)}, j^{(3)} < j$ such that $j^{(1)}, j^{(2)}, j^{(3)} \in R \cup B$ and we have that $j$ overlaps with $j^{(1)}, j^{(2)},$ and $j^{(3)}$.

\item $j \in \overline{\mrk}$ and letting $i \in [k]$ be such that $j \in \overline{\mrk_i}$, we have that $j_i \in R$ (respectively, $j_i \in B$) and there exists $j' > j$ such that $j$ and $j'$ overlap and $j' \in B$ (respectively, $j' \in R$).

\item there exist distinct $j^{(1)}, j^{(2)} > j$ such that $j^{(1)}, j^{(2)} \in R \cup B$ and $j$ overlaps with $j^{(1)}$ and $j^{(2)}$.

\end{enumerate}

Let $\supp(\X) \subseteq [m]$ be the set of all supported chores in $\X$. We also define the set $\usupp(\X) \coloneqq \overline{ R \cup B \cup \supp(\X) }$ and call the chores in $\usupp(\X)$ to be \emph{unsupported} in $\X$. Note that $R$, $B$, $\supp(\X)$, and $\usupp(\X)$ form a partition of the set $[m]$ of all chores.

\end{definition}

Furthermore, we have identified five different exhaustive scenarios in which unsupported chores may arise. For each scenario, we design specific reassignments of chores around the unsupported chores to transition them into supported chores.

Now, we will start identifying the unsupported chores $j$ and make them supported one by one from higher indexed chores to lower index. 
In the process, we will construct a sequence of maximal schedules $\X_i$ for $i \in [k]$. 
For any $ k \geq i \geq 2$, we construct $\X_{k-i+1}$ from $\X_{ k-i}$ as follows: If $\usupp(\X_{ k - i }) \cap \overline{ \mrk_{i } } = \emptyset$, then we set $\X_{k-i+1} = \X_{ k - i }$. Otherwise, let $j = \max ( \usupp(\X_{ k - i }) \cap \overline{ \mrk_{ i } } )$ be the unsupported chore with the largest index in $\X_{ k - i }$, and the following changes are done.

\begin{enumerate}[\text{Case }(1)]
\item If ${j_{i-1}}$ and ${j_{i}}$ overlap, then assign to $j$ the color of ${j_{i-1}}$ and make ${j_{i-1}}$ unassigned~
(see Figure~\ref{condition}-i as example). That is, if ${j_{i}}$ is in $B_{k-i}$ (the other case is analogous) then define 
$$B_{k-i+1} = B_{k-i} \text{ and }$$
$$R_{k-i+1} = R_{k-i} \cup \{j\} \setminus \{{j_{i-1}}\}.$$

\item If ${j_{i-1}}$ and ${j_{i}}$ do not overlap, then we will consider ${j_{i-1}}$ and ${j_{i-2}}$ 
\begin{enumerate}[(a)]
\item If ${j_{i-1}}$ and ${j_{i-2}}$ do not overlap (If ${j_{i-2}}$ is not defined we assume that it does not overlap). Then assign to $j$ the color of ${j_{i-1}}$ and assign to ${j_{i-1}}$ the color of ${j_i}$~
(see Figure~\ref{condition}-ii(a)). That is, if ${j_{i}}$ is in $B_{k-i}$ (the other case is analogous) then define
$$B_{k-i+1} = B_{k-i} \cup \{{j_{i-1}}\} \text{ and }$$
$$R_{k-i+1} = R_{k-i} \cup \{j\} \setminus \{{j_{i-1}}\}.$$

\item If ${j_{i-1}}$ and ${j_{i-2}}$ overlap with each other, consider the following scenarios: 
\begin{enumerate}[(i)]

\item Suppose there exists a chore in $ \overline{\mrk}_i$ which is unsupported in $\X_{k-i}$ and which
does not overlap with any assigned chore with a later finish time. Let ${j}' \in \overline{\mrk}_i$ be such a chore with the latest finish time.
Then, ${j}'$ is assigned the color of ${j_i}$, and ${j_i}$ is assigned the color of ${j_{i-1}}$~
(see Figure~\ref{condition}-ii(b)(i)).
That is, if ${j_{i}}$ is in $B_{k-i}$ (the other case is analogous) then define
$$B_{k-i+1} = B_{k-i} \cup \{{j}'\} \setminus \{{j_{i}}\} \text{ and }$$
$$R_{k-i+1} = R_{k-i} \cup \{{j_{i}}\}.$$

\item If there is no such unsupported chore in $\overline{\mrk}_i$ that meets case~(2)(b)(i), then assign to ${j_i}$ the color of ${j_{i-1}}$~
(see Figure~\ref{condition}-ii(b)(ii)).
That is, if ${j_{i}}$ is in $B_{k-i}$ (the other case is analogous) then define
$$B_{k-i+1} = B_{k-i} \setminus \{{j_i}\} \text{ and }$$
$$R_{k-i+1} = R_{k-i} \cup \{{j_i}\}.$$
We observed that after performing the above reorganization,
there exists a special sub-case where chore assignment in phase 3 may lead to non-maximality. To resolve this issue, we introduce an additional reassignment if such a chore is present in $\overline{\mrk}_{i}$. The special scenario is when a supported chore in $\overline{\mrk}_{i}$ that overlap with exactly three chores ${j_{i-2}}$, ${j_{i-1}}$, and ${j_i}$,
but do not overlap with any assigned chore with a later finish time.
Let ${j}'' \in \overline{\mrk}_i$ be such a chore with the latest finish time. Then we perform reassignment over the case (2)(b)(ii).
Assign to ${j}''$ the color of ${j_{i-2}}$ and make ${j_{i-2}}$ unassigned~
(see Figure~\ref{condition}-special scenario).
That is, if ${j_{i}}$ was in $B_{k-i}$ before performing case(2)(b)(ii) and now ${j_{i}}$ in $R_{k-i+1}$ (the other case is analogous) then define
$$B_{k-i+2} = B_{k-i+1} \cup \{j''\} \setminus \{{j_{i-2}}\} \text{ and }$$
$$R_{k-i+2} = R_{k-i+1}.$$

\end{enumerate}
\end{enumerate}
\end{enumerate}
This finishes the construction of the first part of the desired sequence. From the construction, it
is evident that any two consecutive schedules are adjacent. 
\begin{figure*}[h] 
  \centering\begin{tikzpicture}[scale=0.93]
    \def\rectwidth{15cm}
    \def\rectheight{4.5cm}
    \draw (0,0.5) rectangle (\rectwidth, \rectheight);
    \draw (5,0.5) -- (5, \rectheight);
    \draw (10,0.5) -- (10, \rectheight);
    
    \draw (2.5,8) -- (12.5,8);
    \draw (7.5,4.5) -- (7.5, 8);
    \draw (12.5,4.5) -- (12.5, 8);
    \draw (2.5,4.5) -- (2.5,8);
    \def\startintervals1{{2.8,3.7,3.5,5.5,4.9}}
    \def\endInterval1{{4.1,4.6,5.4,6.7,7}}
    \def\height1{{7.1,6.7,6.2,5.7,5.2}}
    \def\labels1{{"R", "B", "N", "R", "B"}}
    \def\name1{{"$j_{i-1}$", "$j_i$", "$j$", "$j_{i+1}$", "$j_{i+2}$"}}
    \def\colors1{{"red", "blue", "black", "red", "blue"}}
    \foreach \y in {1, 2, 3, 4, 5}
    {
        \pgfmathsetmacro\startx{\startintervals1[\y-1]}
        \pgfmathsetmacro\endx{\endInterval1[\y-1]}
        \pgfmathsetmacro\h{\height1[\y-1]}
        \pgfmathsetmacro\labeltext{\labels1[\y-1]}
        \pgfmathsetmacro\name{\name1[\y-1]}
        \pgfmathsetmacro\labelcolor{\colors1[\y-1]}

        \draw (\startx, \h) -- (\endx, \h);
        \draw (\startx, \h + 0.1) -- (\startx, \h - 0.1);
        \draw (\endx, \h + 0.1) -- (\endx, \h - 0.1);
        
        \node[right, text=\labelcolor] at (\endx - .4, \h+.3) {\labeltext};
        \node[right] at (\startx - .2, \h-.3) {\scriptsize \name};
    }
    \draw (5.4 - 0.35, 6.2 + 0.15) -- (5.4 + 0.05, 6.2 + 0.45);
    \draw (4.1 - 0.35, 7.1 + 0.15) -- (4.1 + 0.02, 7.1 + 0.48);
    \node[right] at (4.1 - 0.4, 7.1 + 0.64 ) {$N$};
    \node[right, text=red] at (5.4 - 0.4, 6.2 + 0.64) {$R$};
    \node[right] at (6.1, 7.75) {Case~(1)};
    
    \def\startintervals1{{7.7,8.8,9.9,9,11.1,10.8}}
    \def\endInterval1{{8.7,9.8,10.5,11,12,12.2}}
    \def\height1{{7.3,6.9,6.5,6.1,5.7,5.1}}
    \def\labels1{{"B", "R", "B", "N", "R", "B"}}
    \def\colors1{{"blue", "red", "blue", "black", "red", "blue"}}
    \def\name1{{"$j_{i-2}$","$j_{i-1}$", "$j_i$","$j$", "$j_{i+1}$", "$j_{i+2}$"}}
    \foreach \y in {1, 2, 3, 4, 5, 6}
    {
        \pgfmathsetmacro\startx{\startintervals1[\y-1]}
        \pgfmathsetmacro\endx{\endInterval1[\y-1]}
        \pgfmathsetmacro\h{\height1[\y-1]}
        \pgfmathsetmacro\labeltext{\labels1[\y-1]}
        \pgfmathsetmacro\labelcolor{\colors1[\y-1]}
        \pgfmathsetmacro\name{\name1[\y-1]}

        \draw (\startx, \h) -- (\endx, \h);
        \draw (\startx, \h + 0.1) -- (\startx, \h - 0.1);
        \draw (\endx, \h + 0.1) -- (\endx, \h - 0.1);
        
        \node[right, text=\labelcolor] at (\endx - .4, \h+.3) {\labeltext};
        \node[right] at (\startx - .2, \h-.3) {\scriptsize \name};
    }
    \draw (11 - 0.35, 6.1 + 0.15) -- (11 + 0.02, 6.1 + 0.48);
    \draw (9.8 - 0.35, 6.9 + 0.15) -- (9.8 + 0.05, 6.9 + 0.45);
    \node[right, text=red] at (11 - 0.4, 6.1 + 0.64 ) {$R$};
    \node[right, text=blue] at (9.8 - 0.4, 6.9 + 0.64) {$B$};
    \node[right] at (10.7, 7.75) {Case~(2)(a)};

    \def\startintervals1{{0.2,0.8,1.7,1.45,1.4,3.7,3.5}}
    \def\endInterval1{{1.15,1.6,2.7,3.2,3.6,4.5,4.8}}
    \def\height1{{3.9,3.5,3.0,2.5,2.0,1.6,1.1}}
    \def\labels1{{"B", "R", "B", "N", "N", "R", "B"}}
    \def\colors1{{"blue", "red", "blue", "black", "black", "red", "blue"}}
    \def\name1{{"$j_{i-2}$","$j_{i-1}$", "$j_i$", "$j'$", "$j$", "$j_{i+1}$", "$j_{i+2}$"}}
    \foreach \y in {1, 2, 3, 4, 5, 6, 7}
    {
        \pgfmathsetmacro\startx{\startintervals1[\y-1]}
        \pgfmathsetmacro\endx{\endInterval1[\y-1]}
        \pgfmathsetmacro\h{\height1[\y-1]}
        \pgfmathsetmacro\labeltext{\labels1[\y-1]}
        \pgfmathsetmacro\labelcolor{\colors1[\y-1]}
        \pgfmathsetmacro\name{\name1[\y-1]}
        
        \draw (\startx, \h) -- (\endx, \h);
        \draw (\startx, \h + 0.1) -- (\startx, \h - 0.1);
        \draw (\endx, \h + 0.1) -- (\endx, \h - 0.1);
        
        \node[right, text=\labelcolor] at (\endx - .4, \h+.25) {\labeltext};
        \node[right] at (\startx - .2, \h-.3) {\scriptsize \name};
    }
    \draw (2.7 - 0.35, 3 + 0.10) -- (2.7 + 0.02, 3 + 0.43);
    \draw (3.2 - 0.35, 2.5 + 0.15) -- (3.2 + 0.05, 2.5 + 0.45);
    \node[right, text=red] at (2.7 - 0.4, 3 + 0.61 ) {$R$};
    \node[right, text=blue] at (3.2 - 0.4, 2.5 + 0.64) {$B$};
    \node[right] at (2.85, 4.15) {Case~(2)(b)(i)};

    \def\startintervals1{{0.2,0.8,1.9,1.4,3.5,3.3}}
    \def\endInterval1{{1.15,1.8,2.9,3.4,4.3,4.5}}
    \def\height1{{3.5,3.1,2.6,2.1,1.7,1.2}}
    \def\labels1{{"B", "R", "B", "N", "R", "B"}}
    \def\colors1{{"blue", "red", "blue", "black", "red", "blue"}}
    \def\name1{{"$j_{i-2}$","$j_{i-1}$", "$j_i$", "$j$", "$j_{i+1}$", "$j_{i+2}$"}}
    \foreach \y in {1, 2, 3, 4, 5, 6}
    {
        \pgfmathsetmacro\startx{\startintervals1[\y-1]+5}
        \pgfmathsetmacro\endx{\endInterval1[\y-1]+5}
        \pgfmathsetmacro\h{\height1[\y-1]}
        \pgfmathsetmacro\labeltext{\labels1[\y-1]}
        \pgfmathsetmacro\labelcolor{\colors1[\y-1]}
        \pgfmathsetmacro\name{\name1[\y-1]}

        \draw (\startx, \h) -- (\endx, \h);
        \draw (\startx, \h + 0.1) -- (\startx, \h - 0.1);
        \draw (\endx, \h + 0.1) -- (\endx, \h - 0.1);
        
        \node[right, text=\labelcolor] at (\endx - .4, \h+.25) {\labeltext};
        \node[right] at (\startx - .2, \h-.3) {\scriptsize \name};
    }
    \draw (7.9 - 0.35, 2.6 + 0.10) -- (7.9 + 0.02, 2.6 + 0.43);
    \node[right, text=red] at (7.9 - 0.4, 2.6 + 0.61 ) {$R$};
    \node[right] at (7.75, 4.15) {Case~(2)(b)(ii)};

    \def\startintervals1{{0.2,0.8,1.7,0.5,1.4,3.7,3.5}}
    \def\endInterval1{{1.15,1.6,2.7,3.2,3.6,4.5,4.8}}
    \def\height1{{3.5,3.1,2.6,2.2,1.7,1.3,0.9}}
    \def\labels1{{"B", "R", "B", "N", "N", "R", "B"}}
    \def\colors1{{"blue", "red", "blue", "black", "black", "red", "blue"}}
    \def\name1{{"$j_{i-2}$","$j_{i-1}$", "$j_i$", "$j''$", "$j$", "$j_{i+1}$", "$j_{i+2}$"}}
    \foreach \y in {1, 2, 3, 4, 5, 6, 7}
    {
        \pgfmathsetmacro\startx{\startintervals1[\y-1] +10}
        \pgfmathsetmacro\endx{\endInterval1[\y-1] +10}
        \pgfmathsetmacro\h{\height1[\y-1]}
        \pgfmathsetmacro\labeltext{\labels1[\y-1]}
        \pgfmathsetmacro\labelcolor{\colors1[\y-1]}
        \pgfmathsetmacro\name{\name1[\y-1]}

        \draw (\startx, \h) -- (\endx, \h);
        \draw (\startx, \h + 0.1) -- (\startx, \h - 0.1);
        \draw (\endx, \h + 0.1) -- (\endx, \h - 0.1);
        
        \node[right, text=\labelcolor] at (\endx - .4, \h+.25) {\labeltext};
        \node[right] at (\startx - .2, \h-.3) {\scriptsize \name};
    }
    \draw (11.15 - 0.35, 3.5 + 0.10) -- (11.15 + 0.02, 3.5 + 0.43);
    \draw (12.7 - 0.35, 2.6 + 0.10) -- (12.7 + 0.02, 2.6 + 0.43);
    \draw (13.2 - 0.35, 2.2 + 0.15) -- (13.2 + 0.05, 2.2 + 0.45);
    \node[right, text=black] at (11.15 - 0.4, 3.5 + 0.61 ) {$N$};
    \node[right, text=red] at (12.7 - 0.4, 2.6 + 0.61 ) {$R$};
    \node[right, text=blue] at (13.2 - 0.4, 2.2 + 0.64) {$B$};
    \node[right] at (12.05, 4.15) {special scenario};

\end{tikzpicture} 
  \caption{Visual illustration of all five scenarios along with the reassignment guidelines for unsupported chores. Letters with a cancel sign represent the color previously assigned to the chore, while the color mentioned above it reflects the assignment after reassignment.}
  \label{condition}
\end{figure*}
Before moving on to phase 3, we prove certain desired properties of the sequence $(\X_0, \X_1, \dots, \X_{k})$. 
\begin{claim}
\label{claim:c-supported}
  For any $i,h$, where $1 \leq i,h \leq k+1 $, if a chore ${j_h}$ $\in$ $\mrk$ is unassigned in the schedule $\X_{k-i+1}$, 
  then ${j_h}$ overlaps with two assigned chores in $\X_{k-i+1}$
  whose finish times are later than ${j_h}$. Hence, ${j_h}$ is supported in $\X_{k-i+1}$.
\end{claim}

\begin{claim}
\label{claim:u-supported}
    For any $i,h$, where $1 \leq i,h \leq k+1$, if the chore ${j_h} \in \mrk$ is not in its source bundle in schedule $\X_{k-i+1}$, then for any $p\geq h$, all the unassigned chores in $\overline{\mrk}_p$ are supported in $\X_{k-i+1}$.
\end{claim}  

\begin{claim}
\label{claim:maximal}
  For any $i$, where $1 \leq i \leq k+1$, the schedule $\X_{k-i+1}$ is feasible and maximal.
\end{claim}

\begin{proof} 
We will prove the three claims through an induction on $i$.

\textbf{Base case}: $i=k+1$. In schedule $\X_0$, each ${j_h}$ is assigned and is in its source bundle by definition. 
Hence, Claims~\ref{claim:c-supported} and \ref{claim:u-supported} are vacuously true. 
And we have already argued the feasibility and maximality of $\X_0$.\\

\textbf{Induction hypothesis}: the three statements are true for the schedule $\X_{k-i}$.\\

\textbf{Induction step}:
Recall that if the schedule $\X_{k-i}$ has no unsupported chores from $\overline{\mrk}_{i}$, then 
$\X_{k-i+1} = \X_{k-i}$. And thus, the three statements will be true for  $\X_{k-i+1}$ as well. 

Now, consider the case when $\X_{k-i}$ has unsupported chores from $\overline{\mrk}_{i}$. 
Then from, induction hypothesis for Claim~\ref{claim:u-supported}, we know that ${j_i}, {j_{i-1}}$, and  ${j_{i-2}}$
are in their source bundles in $\X_{k-i}$.
Also, all chores in $\overline{\mrk}_i$ must be unassigned in $\X_{k-i}$ because it is feasible (from Claim~\ref{claim:maximal}). 

Now, we go over various cases considered in the construction of  $\X_{k-i+1}$.
Recall that $j \in \overline{\mrk}_i$ is the chore that is unsupported in schedule $\X_{k-i}$ and has the latest finish time. 
 
\begin{enumerate}[\text{Case }(1)]
\item ${j_i}$ and ${j_{i-1}}$ overlap each other. In this case, we had the following reassignment: $j \rightarrow color({j_{i-1}})$ and ${j_{i-1}} \rightarrow$ \texttt{unassigned}.\\

For claim~\ref{claim:c-supported}, observe that ${j_{i-1}}$ overlaps with ${j_i}$ and $j$, which are both assigned in $\X_{k-i+1}$. The status of any other ${j_h} $ in $\X_{k-i+1}$ remains same as in $\X_{k-i}$. Hence, the claim follows.\\

For claim~\ref{claim:u-supported}, observe that any chore $j' \in \overline{\mrk}_{i-1}$ overlaps with ${j_i}$ and $j$, and hence satisfies support condition (3). Any chore $j'' \in \overline{\mrk}_i$ having finishing time before $j$ overlaps with ${j_i}$ and $j$ (which are in opposite bundles), and hence satisfies support condition (2). If chore $j'' \in \overline{\mrk}_i$ having finishing time after $j$, it must be supported in $\X_{k-i}$ and will be supported in $\X_{k-i+1}$. If $j''$ is supported due to conditions 2 and 3 then it will supported due to the same condition as we are not changing the bundles for $j_i$ and any $j_h$ where $h\geq i$. if $j''$ is supported due to condition 1 by overlapping with ${j_{i-2}}, {j_{i-1}}, {j_{i}}$, is now supported in $\X_{k-i+1}$ by overlapping with three assigned chores $ {j_{i-2}}, j_i, {j}$.
Any $j''' \in \overline{\mrk}_q$ for $q>i$, if $j'''$ overlaps with $j{i-1}$ in $X_{k-i}$ then in $\X_{k-i+1}$ it will be replaced by assigned chore $j$. 
Any other $j' \in \overline{\mrk}_q$, if it does not overlap with $j_{i-1}$ then there is no change with respect to $j'$. Consequently, $j'$ has the same status in $\X_{k-i+1}$ as $\X_{k-i}$. Hence, the claim follows.\\

For claim~\ref{claim:maximal}, observe that
since $j$ is unsupported in $\X_{k-i}$, 
 we know it does not overlap with any assigned chore with a later finish time, which is in the same bundle as ${j_{i-1}}$.
 Hence, the assignment of $color({j_{i-1}})$ to $j$ is feasible. 
 Leaving ${j_{i-1}}$ unassigned is also maximal, as it overlaps with both bundles: ${j_i}$, which is in the opposite bundle of ${j_{i-1}}$, and 
 $j$, which is in the bundle of ${j_{i-1}}$. 
Any chore $j'\in \overline{\mrk}_{i-1}$ or $j'' \in \overline{\mrk}_i$ other than $j$, 
overlaps with $j$ and ${j_i}$. Hence, their being unassigned is fine. 
 For any other chores, their status w.r.t. maximality are same in $\X_{k-i}$ and $\X_{k-i+1}$.\\

\item ${j_i}$ and ${j_{i-1}}$ do not overlap with each other, and further, we consider ${j_{i-1}}$ and ${j_{i-2}}$.
\begin{enumerate}[\text(a)]
\item ${j_{i-1}}$ and ${j_{i-2}}$ does not overlap. 
We proceeded with the following color reassignment: $j \rightarrow color({j_{i-1}})$ and ${j_{i-1}} \rightarrow color({j_i})$.
Thus, in schedule $\X_{k-i+1}$, $j$ is in the opposite bundle of ${j_i} $ and ${j_{i-1}}$.

Claim~\ref{claim:c-supported} is true because no new ${j_h}$ is unassigned. 

For Claim~\ref{claim:u-supported}, observe that any $j' \in \overline{\mrk}_{i-1}$ is supported because it overlaps with ${j_i}$ and $j$, which are in opposite bundles. 
Any $j'' \in \overline{\mrk}_i$ having finishing time before $j$ overlaps with ${j_i}$ and $j$ (which are in opposite bundles), and hence satisfies support condition (2). If chore $j'' \in \overline{\mrk}_i$ having finishing time after $j$, it must be supported in $\X_{k-i}$ and will be supported in $\X_{k-i+1}$. If $j''$ is supported due to conditions 2 and 3 then it will supported due to the same condition as we are not changing the bundles for $j_i$ and any $j_h$ where $h\geq i$. if $j''$ is supported due to condition 1 by overlapping with ${j_{i-2}}, {j_{i-1}}, {j_{i}}$, is still supported as all three are still assigned in $\X_{k-i+1}$. 
And any $j''' \in \overline{\mrk}_q$ for $q > i$ has the same status in $\X_{k-i+1}$ as in $\X_{k-i}$(similar explanation as in case (1)).\\

For Claim~\ref{claim:maximal}, observe that the
assignment of  ${j_{i-1}}$ in bundle of ${j_i}$ is feasible because ${j_{i-1}}$ does not overlap with ${j_i}$ and ${j_{i-2}}$, and thus, does not overlap 
with any assigned chores. 
Since $j$ is unsupported in $\X_{k-i}$, it does not overlap with any chore with a later finish time which is in the same bundle as ${j_{i-1}}$.
Hence, putting $j$ in the earlier bundle of ${j_{i-1}}$ is fine. 
  As argued in the previous paragraph, any $j' \in \overline{\mrk}_{i-1}$ or $j'' \in \overline{\mrk}_i$ (other than $j$)
  overlap with two chores from opposite bundles. And hence, their being unassigned is fine.\\

\item  ${j_{i-1}}$ and ${j_{i-2}}$ overlap with each other. 
\begin{enumerate}[(i)]
    \item Further there exists an unsupported chore $j' \in \overline{\mrk}_i$ that does not overlap with any assigned chore with a later finish time. We proceeded with the reassignment $j' \rightarrow color({j_i})$ and ${j_i} \rightarrow color({j_{i-1}})$.\\

Claim~\ref{claim:c-supported} is true because no new ${j_h}$ is unassigned.\\

For Claim~\ref{claim:u-supported}, consider any chore $j'' \in \overline{\mrk}_i$. 
If $j''$ has a finish time earlier than $j'$, then it is supported because it overlaps with ${j_i}$ and $j'$ with a later finish time and from the opposite bundle. 
If $j''$ has a finish time later than $j'$, then it is supported because it overlaps with three assigned chores with earlier finish times, ${j_{i-1}}, {j_{i}}, j'$. 
Consider any chore $j''' \in \overline{\mrk}_q$ for $q > i$. If $j'''$ was supported in $\X_{k-i}$ due to its overlap with ${j_i}$, now in $\X_{k-i+1}$,
$j'$ can take place of ${j_i}$.\\

 For Claim~\ref{claim:maximal}, observe that $j'$ is unsupported in $\X_{k-i}$, hence only overlaps with two assigned chores ${j_{i-1}}$ and ${j_i}$ from opposite bundles. And ${j_i}$ does not overlap with ${j_{i-1}}$. Thus, the reassignment keeps the schedule feasible. 
  Any chore $j'' \in \overline{\mrk}_{i-1}$ or $j''' \in \overline{\mrk}_{i}$ (other than $j'$) overlaps with ${j_{i-1}}$ and $j'$ from opposite bundles, and hence can remain unassigned.\\

\item Chores ${j_i}$, ${j_{i-1}}$, and ${j_{i-2}}$ are positioned like case~(2)(b)(i).
Additionally, there is no unsupported $j' \in \overline{\mrk}_i$ that does not overlap with any assigned chore with a later finish time. 
We proceeded with the color reassignment: ${j_{i}} \rightarrow color({j_{i-1}})$.\\

Claim~\ref{claim:c-supported} is true because no new ${j_h}$ is unassigned.\\

For Claim~\ref{claim:u-supported}, consider any chore $j'' \in \overline{\mrk}_i$. 
If $j''$ was supported in $\X_{k-i}$ by overlapping with three assigned chores with earlier finish time,
then that will also hold true in $\X_{k-i+1}$ as all assigned chores are still assigned. 
Consider $j''' \in \overline{\mrk}_i$ which is unsupported in $\X_{k-i}$ and overlaps with an assigned chore with a later finish time, say $j_h$. 
Because $j'''$ is unsupported, $j_h$ and ${j_{i}}$ must be in the same bundle in $X_{k-i}$ and $j'''$
must not overlap any other chore from a later finish time and in the opposite bundle. 
We can conclude that any unsupported chore $j''' \in \overline{\mrk}_i$ must overlap with $j_h$. 
In schedule $X_{k-i+1}$, $j_h$ and ${j_i}$ are in opposite bundles, and hence, $j'''$ is supported. 
Similarly, any chore $j''' \in \overline{\mrk}_{i+1}$ overlaps with $j_h$ and ${j_i}$, and hence, is supported.\\

For Claim~\ref{claim:maximal}, observe that in schedule $\X_{k-i}$, the chore ${j_{i}}$ neither overlaps with ${j_{i-1}}$,
nor with any other chore in the same bundle. Hence, it is feasible to put it in the bundle of ${j_{i-1}}$.
As argued in the previous paragraph, any $j'' \in \overline{\mrk}_i$ overlaps with two chores in opposite bundles and, hence, can be left unassigned.\\

 Further, after performing case (2)(b)(ii) if the special scenario arise where there is a chore $j'' \in \overline{\mrk}_i$ that overlaps with chores ${j_{i-2}}, {j_{i-1}}, {j_{i}}$,
but does not overlap with any assigned chore with a later finish time.
Then we reassigned colors as follows: $j'' \rightarrow color({j_{i-2}})$ and ${j_{i-2}} \rightarrow$ \texttt{unassigned}.\\

For Claim~\ref{claim:c-supported}, observe that ${j_{i-2}}$ is the newly unassigned chore and it
overlaps with ${j_{i-1}}$ and $j''$ from opposite bundles. Hence, ${j_{i-2}}$ is supported.\\

For Claim~\ref{claim:u-supported}, observe that any chore $j''' \in \overline{\mrk}_{i-2}$ overlaps with ${j_{i-1}}$ and $j''$, which have later finish times, and hence is supported.
Any chore $j''' \in \overline{\mrk}_{i-1}$ overlaps with ${j_{i-1}}$ and $j''$ from opposite bundles, and hence, is supported. 
Any chore $j''' \in \overline{\mrk}_{i}$ whose finish time is between ${j_i}$ and $j''$ is similarly supported by overlapping with ${j_i}$ and $j''$.
Any chore $j''' \in \overline{\mrk}_i$ whose finish time is later than $j''$ is supported because it overlaps with assigned chores ${j_{i-1}}, {j_i}, j''$
with earlier finish times. 
For any chore $j''' \in \overline{\mrk}_q$ for $q > i$, 
if it was supported in $\X_{k-i}$ by overlap with ${j_i}$, then $j''$ can take the place of ${j_i}$.\\

For Claim~\ref{claim:maximal}, observe that since ${j_i}$ has been given $color({j_{i-1}})$ and ${j_{i-2}}$ is unassigned,
it is feasible to give $color({j_{i-2}})$ to $j''$.
Similarly, we can leave ${j_{i-2}}$ unassigned because it overlaps with $j''$ and ${j_{i-1}}$, which are in opposite bundles.
For any unassigned $j''' \in \overline{\mrk}_{i-2}, \overline{\mrk}_{i-1}, \overline{\mrk}_{i}$, we have argued in the previous paragraph that they overlap with two chores in opposite bundles. 

    \end{enumerate}
\end{enumerate}
\end{enumerate}
\end{proof}

 From the three claims, it is clear that at the end of phase 2, there will be no unsupported unassigned chores.

Following the reassignment, several changes occur: some assigned chores become unassigned, some unassigned chores get assigned, and certain assigned chores achieve their target bundles. We label the chores that are in the same bundle in $\X_k$ as their target bundles and 
the chores that are unassigned in both $\X_k$ and $\X_l$ as 
 ``settled chores". All other chores are labeled as ``unsettled chores".
Before proceeding to phase 3, it is essential to establish one more claim, which will assist us in verifying the feasibility of constructed schedules in phase 3.
\begin{claim}
\label{claim:feasible}
    Let ${j_i} \in \mrk$ be an unsettled chore and let $\overline{j}$ be another chore with a later finish time, which is in the target bundle of ${j_i}$ in schedule $\X_k$.  
    If $\overline{j}$ overlaps with ${j_i}$, then $\overline{j}$ must be unsettled and must be the immediately next unsettled chore after ${j_i}$ (in order of finish times).
\end{claim}
\begin{proof}
    Assuming there are no unsupported chores in $\X_0$, this implies that no reassignment of chores happens and we have $\X_k = \X_0$. It follows that all unassigned chores are settled, while all assigned chores are unsettled. In the $\X_0$ schedule, the only assigned chore with a later finish time that overlaps with ${j_i}$ is ${j_{i+1}}$.

However, in cases involving unsupported chores, color reassignments are performed. We will verify the claim for each individual case one by one. 
\begin{enumerate}

    \item The case where ${j_i}$ and ${j_{i-1}}$ overlap with each other, we perform the following color reassignments: ${j_{i-1}}$ is left unassigned, and $j$ is assigned the color of ${j_{i-1}}$. In this scenario, the unsettled marked chores for which something changed are $j_{i-2}$, ${j_{i-1}}$and, $j_i$. Where $j_{i-2}$ overlaps with ${j_{i-1}}$ that is unassigned, other than this it may overlap some unassigned supported chores only. So $j_{i-2}$ posses no issue. Further $j_{i-1}$ overlaps with $j_i$ (of the target color), but it will be the immediately next unsettled chore because any unassigned chore in $\overline{\mrk}_{i-1}$ retains its status, and remains unassigned. Therefore, the immediately next unsettled chore after ${j_{i-1}}$ is ${j_{i}}$. Furthermore, $j$ becomes assigned but with a color that is not the target color of ${j_{i-1}}$. Importantly, $j$ does not overlap with any assigned chore other than ${j_i}$ and ${j_{i-1}}$ from earlier finishing chores. Additionally, ${j_i}$ gets an unsettled chore immediately after it. Since $j$ does not overlap with ${j_{i+1}}$, it follows that ${j_i}$ also does not overlap with ${j_{i+1}}$. Therefore, the unsettled chore immediately after ${j_i}$ is only $j$. Thus, the claim holds for this case.

    \item In the case where 
    ${j_{i}}$ and ${j_{i-1}}$ are not overlapping with each other and furthermore, ${j_{i-1}}$ does not overlap with ${j_{i-2}}$, the situation is similar to the first case. $j_{i-2}$ does not overlap with any assigned chores from later finish time. None of the chores from $\overline{\mrk}_{i-1}$ change their status; they all remain settled. Since ${j_{i-1}}$ is assigned the target color, it is also considered settled. The only unsettled marked chore is ${j_i}$ for which something changed, which overlaps with $j$ (of the target color). However, $j$ must be an immediately next unsettled chore since every chore between these two becomes supported after this assignment, and, as per the definition of unsupported $j$, ${j_i}$ must not overlap with ${j_{i+1}}$. This confirms the claim being true for this case.   
 
    \item In the last case where 
    ${j_i}$ and ${j_{i-1}}$  are not overlapping each other and ${j_{i-1}}$ overlaps with ${j_{i-2}}$. In the first case of 2(b)(i) the chores for which something changed are $j_{i-1}$ and, $j_i$. $j_{i-1}$ overlaps with $j$ having target color, so it should be immediately next unsettled chore for $j_{i-1}$. This is true because chores from $\overline{\mrk}_{i-1}$ and $\overline{\mrk}_i$ retains its status, and remains unassigned. Chore $j_i$ gets its target color and is settled. Hence, the claim holds for this case. In the second case, 2(b)(ii) the situation changed for only $j_i$ and it became settled. For any other unsettled chores situations are unchanged so the claim will be true for this case. In the last case of the special scenario, it is analogous to the first case when ${j_i}$ and ${j_{i-1}}$ overlap each other. We can observe the claim in the same way as the first case. This completes the proof that, after phase 2, every unsettled assigned chore overlaps only with the immediately next unsettled chore of the target color from the later finish time.
\end{enumerate}
    
\end{proof}

\paragraph{\textbf{Phase 3.}} In this phase, we will construct the second part of the desired sequence of maximal schedules, starting from $\X_k$.
This will be done by assigning the target bundle to the unsettled chores. We will start by identifying the unsettled chores and assign them to their target bundle one by one, from lower-indexed chores to higher-indexed ones. In the process, we will construct a sequence of maximal schedules 
$\X_i$ from  $\X_k$ to $\X_l$.
For any $0 \leq i \leq l-k-1$, we construct $\X_{k+i+1}$ from $\X_{k+i}$ as follows.
If there is no unsettled chore in schedule $\X_{k+i}$, then we are done.
Otherwise, let $c$ be the earliest finishing unsettled chore in schedule $\X_{k+i}$. Then assign $c$ to the target bundle and assign a feasible bundle to the next unsettled chore. A feasible bundle here means that if the chore can get assigned without violating any feasibility condition, then assign that bundle otherwise make it unassigned.

Since we assign at least one chore to its target bundle in each step, it follows that the process will conclude within a linear number of steps.  

\emph{Feasibility and maximality in Phase 3:} Now, we shall establish the feasibility and maximality of schedules in Phase 3. After completing Phase 2, we obtain a schedule $\X_k$ that is both feasible and maximal. In Phase 3, we identify the unsettled chore $c$ from the earliest finish time and reassign its bundle. Further, we assign a feasible bundle to their immediately next unsettled chore.

It is important to note that $c$ never belongs to $\overline{\mrk}$ chores because $c$ is the earliest unsettled chore, and it is from $\overline{\mrk}$, which means it must be assigned ( because if it were unassigned, it would be settled ). Further $c$ also overlaps with two marked chores from earlier finish time ( because $c \in \overline{\mrk} $). These chores are settled and must belong to the opposite bundle. This would violate the feasibility of $c$. 
On the other hand, if $c \in {\mrk}$, assigning it a target bundle is feasible, as per the claim ~\ref{claim:feasible}. This claim established that every unsettled marked chore overlaps only with its immediately next unsettled chore, which can share the same bundle as its target bundle. Since phase 3 involves reassigning bundles to the earliest finish time unsettled chores, and the immediate next chore that can have the same bundle as the target bundle of $c$, this ensures the feasibility of assigning the target bundle to $c$. Assigning a bundle to the immediately next unsettled chore poses no issues, as we can assign any feasible bundle that meets the criteria.

Now, we must demonstrate that all unassigned chores will satisfy the maximality criteria after the reassignments. As we show after Phase 2, every unassigned chore will be supported. That means all unassigned chores must follow one of the three criteria of being supported. Further, we will demonstrate for each case how it will satisfy maximality criteria.
\begin{enumerate}
    \item In the first scenario, let an unassigned chore $c$ be supported due to the third condition, where $c$ overlaps with two assigned chores whose finish times are later than $c$, it is apparent that both assigned chores must belong to different bundles. This distinction is necessary since they must overlap with each other. In phase 3, our color reassignment is from lower indexed chores to higher indexed chores, and thus, we do not alter the status of these particular assigned chores until all the earlier finish time chores are set to their target bundle. In this situation, if $c \in \overline{\mrk}_i$, it must overlap with at least two marked chores $j_i$ and,$j_{i-1}$. So before altering the color of two specific assigned chores from a later finish time, we make $j_i$ and, $j_{i-1}$ settled(They must be in different bundles). That fulfills the maximality requirement of $c$. Conversely, if $c \in \mrk$, it will get assigned its target bundle and will not require maximality.
    
    \item In the second scenario, when an unassigned chore $c$ is supported due to the second condition. Where $c \in \overline{\mrk}_i$ and overlaps ${j_i}$ from bundle $R$ ( respectively $B$ ) and also overlaps with one assigned chore $j'$ from a later finish time from bundle $B$ ( respectively $R$ ). As $c \in \overline{\mrk}$ it also overlap with $j_{i-1}$ and $j_i$. The maximality of $c$ is due to $j_i$ and $j'$. Now if $j_i$ is already settled, we will not reassign its bundle. Before reassigning the bundle of $j'$, both $j_i$ and $j_{i-1}$ will be settled, which ensures the requirement of maximality. If $j_i$ is unsettled then the bundle of $j_i$ and the target bundle of $j_{i-1}$ are the same. Therefore before reassignment of $j_i$, $j_{i-1}$ will get settled that means $j_{i-1}$ will replace $j_i$. Further, before $j'$, $j_i$ will get settled, So $j_i$ and $j_i-1$ will ensure the maximality of $c$.

    \item In the final scenario, when an unassigned chore $c$ is supported solely due to condition 1, it implies that conditions 2 and 3 are not satisfied. $c$ overlaps with three or more assigned chores whose finish time is earlier than the $c$. If $c$ also meets any of the other conditions of supported chores then we already discussed how $c$ satisfies maximality criteria due to other conditions. We can consider the simplest case where the last three assigned chores are from $\mrk$, and they are assigned to their source bundles in $\X_k$, which means they have alternating assigned bundles. In this scenario, it is straightforward to confirm that during each step of phase 3, $c$ will overlap with two assigned chores from different bundles. Now, let's consider another scenario in which the last three assigned chores are from $\mrk$, but their colors are modified, signifying that they do not belong to alternating bundles. One can easily verify that this scenario arises due to the reassignment in case 2(b)(ii) only, where we reassign the bundle of ${j_i} \in \mrk$ but do not assign a color to $j \in \overline{\mrk}_i$. In any other case, we assign a color to either $j \in \overline{\mrk}_i$ or $j' \in \overline{\mrk}_i$. In the scenario in which the last three assigned chores are from $\mrk$, but their colors are modified, we may encounter different color sequences of the three assigned chores. Observing the structure carefully, we see that no supported unassigned chore due to the overlap of three chores with sequences of $B, B, B$, or $R, R, R$. This is because, in such a situation, $j$ will lose its feasibility. However, we have already demonstrated that we consistently maintain feasibility in phase 2. Therefore, the only situation where an issue arises is when the color sequence of the assigned chores is $B, R, R$ and $R, B, B$. This is the exact special scenario that we consider after the reassignment of case(2)(b)(ii). In this case, we reassign the bundle of such a supported unassigned chore.\\
    Moving on, another scenario arises when at least one of the three assigned chores belongs to set $\overline{\mrk}$. Let's assume this specific assigned chore is ${p}' \in \overline{\mrk}_i$. This reassignment can arise from any case except 2(b)(ii). Furthermore, we will elucidate how maximality is upheld in each case. If the reassignment is due to case 1, if the unassigned chore $c \in \overline{\mrk}_{i+2}$ then it will overlap with assigned ${j_i}$, $p'$,$j_{i+1}$ and, $j_{i+2}$. Similarly if the unassigned chore $c \in \overline{\mrk}_{i+1}$ then it will overlap with unassigned ${j_{i-1}}$, assigned $j_i$, $p'$ and,$j_{i+1}$. One can verify that in both scenarios during phase 3, two of these four chores must be assigned to the opposite bundles, ensuring the maximality condition for the unassigned chore $c$. If the reassignment is due to case 2(a), the scenario closely resembles case 1. A similar analysis indicates that two of these overlapping chores must consistently be assigned to opposite bundles during phase 3, which ensures the maximality condition of $c$.

    If the reassignment is due to case 2(b)(i), and $c \in \overline{\mrk}_i$ having finish time later than $p'$. In this scenario $c$ overlap with $j_{i-1}$, $j_i$ and, $p'$. Where $j_i$ is in its target bundle and $p'$ will be in the opposite bundle of $j_i$, that means maximality of $c$ hold due to $j_i$ and $p'$. So in phase 3 before assigning the target bundle to $p'$ we will assign the target bundle to $j_{i-1}$ that is the opposite bundle of $j_i$. So now maximality of $c$ holds due to $j_{i-1}$ and $j_i$. Further, if $c$ belongs to any other unmarked set then the scenario is similar to case 1 or case 2(a). So a similar analysis will hold.

     Finally, in the case of the special scenario of 2(b)(ii), this case is very similar to 2(b)(i). So a similar analysis reveals that the maximality of $c$ is maintained in this case as well. With this, we have covered all sub-cases of unassigned chores that overlap with three or more chores.
    
\end{enumerate}
This concludes the verification of the schedule's feasibility and maximality across all steps.\\ 

\textbf{\emph{Adjacent:}} After completing phase 1, we obtain a feasible and maximal schedule. In each step of phase 2, we reassign colors based on four different cases and a special case. It's worth noting that each of these cases is adjacent to the previously assigned schedule. Moving on to phase 3, in each step, we reassign bundles to only two chores. If these two chores have different bundles, then the reassigned schedules will also be adjacent to the previous schedule.

Now, two adjacent chores cannot reassigned to the same bundle because if they overlap, assigning the same bundle to both is not feasible. If they do not overlap, there is no need to reassign the bundle to the second chore, as the initially assigned bundle would be feasible for it. This ensures that each phase of the steps is adjacent to the preceding one, leading to the conclusion that each $\X_i$ is adjacent to $\X_{i+1}$. Therefore, the sequence of schedules is adjacent to one another.
\end{proof}

\clearpage
\end{document}